\documentclass[10pt,journal,compsoc]{IEEEtran}

\ifCLASSOPTIONcompsoc
  \usepackage[nocompress]{cite}
\else
  \usepackage{cite}
\fi
\usepackage{hyperref}

\ifCLASSINFOpdf
   \usepackage[pdftex]{graphicx}
\else
\fi
\usepackage{amsmath}
\usepackage{amsthm,amsfonts,amssymb}
\ifCLASSOPTIONcompsoc
  \usepackage[caption=false,font=footnotesize,labelfont=sf,textfont=sf]{subfig}
\else
  \usepackage[caption=false,font=footnotesize]{subfig}
\fi
\usepackage{caption}
\usepackage{fixltx2e}

\usepackage{stfloats}
\usepackage{setspace}
\usepackage{multirow}
\usepackage{bm}
\usepackage{capt-of}
\usepackage{color}
\usepackage[all]{xy}

\usepackage{enumitem}
\usepackage{bbold}
\usepackage{float}
\usepackage[dvipsnames]{xcolor}

\newcommand{\ew}[1]{#1}

\newtheorem{theorem}{Theorem}[section]
\newtheorem{corollary}{Corollary}[section]
\newtheorem{lemma}{Lemma}[section]
\newtheorem{definition}{Definition}
\newtheorem{property}{Property}[section]

\newtheorem{assumption}{Assumption}

\let\originaleqref\eqref
\renewcommand{\eqref}{Eq.~\originaleqref}
\newcommand{\req}[1]{(\ref{#1})}

\begin{document}
%
\title{A General Sensitivity Analysis Approach for Demand Response Optimizations}
%
%
%
%

\author{Ding~Xiang,~\IEEEmembership{Student Member,~IEEE,}
        and~Ermin~Wei,~\IEEEmembership{Member,~IEEE,}
\IEEEcompsocitemizethanks{\IEEEcompsocthanksitem Ding Xiang and Ermin Wei are with the Department
of Electrical Engineering and Computer Science, Northwestern University, Evanston,
IL, USA, 60208. \protect 
E-mail: dingxiang2015@u.northwestern.edu,
ermin.wei@northwestern.edu
}
}

\IEEEtitleabstractindextext{%
\begin{abstract}
It is well-known that demand response can improve the system efficiency as well as lower consumers' (prosumers') electricity bills. However, it is not clear how we can either qualitatively identify the prosumer with the most impact potential or quantitatively estimate each prosumer's contribution to the total social welfare improvement when additional resource capacity/flexibility is introduced to the system with demand response, such as allowing net-selling behavior. In this work, we build upon existing literature on the electricity market, which consists of price-taking prosumers each with various appliances, an electric utility company and a social welfare optimizing distribution system operator, to design a general sensitivity analysis approach (GSAA) that can estimate the potential of each consumer's contribution to the social welfare when given more resource capacity. GSAA is based on existence of an efficient competitive equilibrium, which we establish in the paper. When prosumers' utility functions are quadratic, GSAA can give closed forms characterization on social welfare improvement based on duality analysis. Furthermore, we extend GSAA to a general convex settings, i.e., utility functions with strong convexity and Lipschitz continuous gradient. Even without knowing the specific forms the utility functions, we can derive upper and lower bounds of the social welfare improvement potential of each prosumer, when extra resource is introduced. For both settings, several applications and numerical examples are provided: including extending AC comfort zone, ability of EV to discharge and net selling. The estimation results show that GSAA can be used to decide how to allocate potentially limited market resources in the most impactful way.
\end{abstract}

\begin{IEEEkeywords}
Demand response, competitive equilibrium, optimization duality, sensitivity analysis, utility function, strong convexity, Lipschitz continuous gradient.
\end{IEEEkeywords}}

\maketitle

\IEEEdisplaynontitleabstractindextext

%
\IEEEpeerreviewmaketitle

\ifCLASSOPTIONcompsoc
\IEEEraisesectionheading{\section{Introduction}\label{sec:introduction}}
\else
\section{Introduction}
\label{sec:introduction}
\fi

%
%
%
%
\IEEEPARstart{A}{s} opposed to the traditional supply-follow-demand approach, \textit{demand response (DR)} provides consumers an opportunity to \ew{ balance the supply and demand} of electricity systems by responding to time-based pricing signals \cite{5628271, 7572161, 7587952, 6425944, 6161320, 7565746, 32577, 7926413}. Appropriately designed pricing signals can encourage consumers to move their flexible demand from the high-demand periods to relatively low-demand periods \cite{6039082, Bilevel, 7102747}, or also encourage \textit{prosumers}, i.e., \ew{who can both produce and consume} \cite{6480096}, to store the electricity into their energy storage or EVs at low pricing periods and use \ew{energy} in high pricing periods \cite{6948246}, or even sell/discharge it back to the system \ew{to obtain} economic reward \cite{6782392, TOUV2G}. By encouraging the consumers' or prosumers' (for simplicity, we call both of them prosumers in the rest of the paper) participation, DR can significantly lower the system load fluctuation, increase the system efficiency, reduce consumers' electricity bills and thus improve the total social welfare \cite{TwoMarket, 7131795}. 

Even though \ew{DR} has huge potential social benefits, the practical implementation still faces a lot of challenges \cite{OCONNELL2014686, NOLAN20151}. For hardware, the implementation could require upgrade or installation of lines, smart meters or other power devices. For software, the implementation may require well-designed algorithms to calculate the optimal pricing signal or control strategy that can manage flexible appliances. For \ew{the market environment}, it may require utility companies or aggregators to make advertisement or campaign persuading customers to join DR project. However, to make all of these conditions satisfied \ew{requires} huge amount of investment, whereas the budgets are usually limited \cite{Budget4, Budget1, Budget3}.
 
\ew{In this paper, we aim to address the question of how  we can use the limited budget to most efficiently expand the current resource capacity, such as the total amount of allowed net-selling.} One  \ew{solution is to identify and} prioritize prosumers with highest \ew{potential impact on social welfare}, and invest \ew{in }them first in order to maximize the social welfare improvement. There are some related studies on prosumers' marginal contribution in DR. Most of them focus on the flexibility or price-responsiveness of a certain type of \ew{appliances} \cite{FAGHIH2013472, 6601729, 8309117}. Some of them study the DR resources' capacity value based on simulations or empirical analysis \cite{6939174, 8085972}. However, to the best of our knowledge, there is no related studies providing a general tool quantifying or estimating different prosumers' \ew{potential impact on social welfare} under a dynamic pricing environment\ew{, which  is the problem the paper provides solutions to}. A key to this problem is to find \ew{connections between prosumers' utility function/their usage behaviors and the social welfare}.

To find \ew{this connection,} we build	 upon the existing literature on the dynamic pricing DR market framework \cite{6039082, 6161320, 8309117} \ew{to model a market with} price-taking prosumers with various appliances, a price-taking electric utility company and a \ew{price-setting social welfare maximizing} distribution system operator (DSO).  In the market, each prosumer maximizes their own payoff function with considering flexible appliances' energy consumption constraints; the utility company \ew{maximizes} its profit, and the DSO sets appropriate dynamic pricing signals to maximize the social welfare and clears the market. We first \ew{model all different types of prosumers' constraints by a general linear constraint}. Then we show there exists an efficient competitive equilibrium and an equivalence relationship between DSO's problem and the prosumers' and the utility company's \ew{problems}. Based on the general linear constraint and the equivalence relationship, we propose a general sensitivity analysis approach (GSAA) to quantify the effect of prosumers' contribution on social welfare \ew{when more resource is injected into the system, modeled by enlarging their constraint sets}. 

Since each prosumer's usage behavior is highly related to their own preference, which is commonly modeled by a (net) utility function \cite{6480096, 6039082, 8309117, 7778740, 6266724}, the characteristics of different types of utility functions could affect the connection between prosumers' usage behaviors and the social welfare, thus could also change the way we quantify or estimate each prosumer's contribution on social welfare. In this work, we consider two types of (net) utility function settings, quadratic function, which is commonly used in many DR related papers \cite{6039082, 8309117, 7778740, 6266724} and general convex function (i.e., strongly convex and with Lipschitz continuous gradient), which is more \ew{general and captures a variety of functions.}

As the main contributions of this paper, we \ew{relate the potential social welfare improvement and individual prosumer by duality theory and use the proposed GSAA to characterize  the shadow prices associated with more resources in the system, i.e., a larger constraint set. When the net utility functions have quadratic forms, we derive the shadow prices explicitly and when the utility functions have general convex form, we establish the bounds of the shadow prices.} Under different utility settings, several applications of GSAA are provided, including enlarging a prosumer's AC comfort zone size, allowing a prosumer's EV to discharge and allowing a prosumer to net sell. \ew{The} estimation can be \ew{used to identify the  prosumers with the most potential impact on social welfare}. Thus, when the budgets for implementing DR are limited, we could allocate them to \ew{those prosumers accordingly}. We also provide a case study in the general convex settings to compare 2 prosumers' contribution potentials for allowing them to discharge EVs. \ew{The two prosumers are identical,} except that one's net utility function is the other's scaled by a constant. We would expect the one with \ew{a} larger utility function \ew{to} have a higher contribution potential, but \ew{to our surprise,} the bounds of the shadow price given by GSAA \ew{show otherwise under certain conditions}. \ew{Lastly}, several numerical studies about the usage of GSAA in net-selling and EV discharging are provided.

The rest of the paper is organized as follows. In Section 2, we introduce the system model and the general linear constraint. In Section 3, we prove the existence of the efficient market equilibrium and the equivalence between DSO problem and \ew{prosumers' and the utility companies'} problems, which form the foundation of GSAA. Section 4 studies the usage of GSAA under quadratic net utility function settings, where the closed-form shadow price are calculated and Section 5 talks about the usage of GSAA under the general convex settings in which the bounds of the shadow price are derived. In Sections 6, we provide numerical studies to illustrate the usage of GSAA.

\noindent\textbf{Notations:} Throughout the paper, we use bold font to represent vectors variables,  and superscript $H$ for transpose. Notation $\left( x_{i}, \forall  i\right)$ denotes vector $\left[ x_{1}, x_{2}, ..., x_{n}\right]^{T} $. For column vectors $\textbf{x}_1, \textbf{x}_2, ..., \textbf{x}_n$, we use notation $\left( \textbf{x}_{i}, \forall  i\right)$ to denote a combined column vector $\left[ \textbf{x}_{1}^T, \textbf{x}_{2}^T, ..., \textbf{x}_{n}^T\right]^{T}$. For a vector $\textbf{x}=\left[x_1, x_2, ..., x_n\right]^T$, we denote the $j^{th}$ element of $\textbf{x}$ as $[\textbf{x}]_j$.The operator $[\cdot]^+$ is given by $[a]^+ = \max\{a,0\}$. We use $\|\cdot\|$ to denote that Euclidean norm or $L_2$ norm. For a matrix $A$, we denote its element on the $i^{th}$ row and the $j^{th}$ column as $\left[A\right]_{ij}$. Inequalities, $<$, $\leq$, $>$, $\geq$ are used in an element-wise sense. The notation $diag\{A\}$ denotes a column vector consists of the diagonal elements of matrix A. We use notation $diag\left(A, B\right)$ to denote matrix
$\begin{bmatrix}
    A       &  0 \\
    0       &  B
\end{bmatrix}$.
For a function $f(x_1, x_2, ..., x_n): \mathbb{R}^{n}\rightarrow \mathbb{R}$, $\nabla f$ is the gradient of $f$, and $\nabla f=(\frac{\partial f}{\partial x_1}, \frac{\partial f}{\partial x_2}, ..., \frac{\partial f}{\partial x_n})^T$.

\section{System Model}\label{sec:Systemmodel}

In this section, we introduce our general market model consisting of three participants: prosumers, an electric utility company and a DSO. We use $\mathcal{N}=\{0,1,2,...,N\}$ to denote the set of all nodes in the power network. Without loss of generality, we let node 0 represent the electric utility company and the nodes in $\mathcal{N}^{+} = \{1,2,...,N\}$ represent individual prosumers. In rest of the paper, the words `prosumer node' and `prosumer' are used interchangeably. Similarly, we do not distinguish between the words `electric utility company node', `electric utility company'  or `utility company' unless otherwise noted. The DSO, as a social planner and a system operator, maximizes the social welfare while clearing the market and guarantees system operation reliability. The market we consider here is a retail electricity market with no uncertainty and $H$ periods, indexed by $t \in \mathcal{H}=\{1,2,...,H\}$. We assume in this market both prosumers and the electric utility company are price-taking. We call the market $H$-period DR market.


For each prosumer $i \in \mathcal{N}^{+}$, let $\mathcal{A}_{i}$ denote the set of household appliances, such as lights, air conditioners (ACs), washers, energy storage and electric vehicles (EVs). For each appliance $a \in \mathcal{A}_{i}$, we define its \textit{power consumption scheduling vector} as
$\textbf{q}_{i,a}^D$ $\overset{\Delta}{=}$ $\left(q_{i,a}^D\left(t\right), \forall  t\right)$ in $\mathbb{R}^H$ where $q_{i,a}\left(t\right) \in \mathbb{R}$ represents the power consumption of prosumer $i$'s appliance $a$ at time $t$, and superscript $D$ indicates the demand side (i.e., the prosumers). We can concatenate these vectors to define for each prosumer $i$, $\textbf{q}_i^D$ $\overset{\Delta}{=}$ $\left(\textbf{q}_{i,a}^D, \forall  a\right)$ in $\mathbb{R}^{H|\mathcal{A}_{i}|}$ and the aggregate vector for all prosumers,
$\textbf{q}^D$ $\overset{\Delta}{=}$ $\left(\textbf{q}_i^D, \forall  i\right)$ in $\mathbb{R}^{H\sum_{i \in \mathcal{N}^+}|\mathcal{A}_{i}|}$. For the utility company (superscript $S$ representing the supply side), we define its \textit{power supply scheduling vector} as $\textbf{q}^S$, where $\textbf{q}^S\overset{\Delta}{=}$ $\left(q^S\left(t\right), \forall  t\right)$ in $\mathbb{R}^H$. At last, we define the \textit{market clearing electricity price vector}, which is announced by DSO, as $\textbf{p}$ $\overset{\Delta}{=}$ $\left(p\left(t\right), \forall  t\right)$ in $\mathbb{R}^{H}$. 

\subsection{ Prosumer Model} 
In our model, prosumers (i.e., users who can both buy and sell back electricity) optimize their power schedules to maximize their own payoffs. Here the payoff function is defined by a quasilinear function $U_i\left(\textbf{q}_i^D\right)-C_i\left(\textbf{q}_i^D\right)- \sum_{a  \in\mathcal{A}_{i}}\textbf{p}^{T}\textbf{q}_{i,a}^D$, which includes a utility function $U_i: \mathbb{R}^{H |\mathcal{A}_{i}|} \rightarrow \mathbb{R}$ that reflects prosumer $i$'s preference on energy consumption for appliances, a cost function $C_i: \mathbb{R}^{H |\mathcal{A}_{i}|} \rightarrow \mathbb{R}$ representing amortized battery life loss for EVs and energy storage (for other non-battery appliances, $C_i = 0$), and payments $\sum_{a  \in\mathcal{A}_{i}}\textbf{p}^{T}\textbf{q}_{i,a}^D$. The combination terms $U_i\left(\textbf{q}_i^D\right)-C_i\left(\textbf{q}_i^D\right)$ are also called \textit{net utility function}. With these components, the prosumers' problem (P) is formulated as follows. For each prosumer $i \in \mathcal{N}^{+}$ 
\begin{align} \label{opt:Prosumer}
\max_{\textbf{q}_i^D}  \quad
&U_i\left(\textbf{q}_i^D\right)-C_i\left(\textbf{q}_i^D\right)- \sum_{a  \in\mathcal{A}_{i}}\textbf{p}^{T}\textbf{q}_{i,a}^D,\tag{P1} \\
\text{s.t.}  \quad 
&A_{i}\textbf{q}_i^D \leq \textbf{h}_{i},\tag{P2}\label{pro m2} 
\end{align}
We adopt the following simplifying assumption:
\begin{assumption}\label{separable}
Prosumer $i$'s utility function $U_i(\textbf{q}_i^D)$ and cost function $C_i(\textbf{q}_i^D)$ are separable over appliances and time periods.
\end{assumption} 
Having the above assumption, prosumer $i$'s utility function and cost function can be written as 
\begin{gather}
U_i(\textbf{q}_i^D)=\sum_{a \in \mathcal{A}_{i}}\sum_{t \in \mathcal{H}}U_{i,a}^t(q_{i,a}^D\left(t\right)),\\
C_i(\textbf{q}_i^D)=\sum_{a \in \mathcal{A}_{i}}\sum_{t \in \mathcal{H}}C_{i,a}^t(q_{i,a}^D\left(t\right)),
\end{gather}
where $U_{i,a}^t: \mathbb{R} \rightarrow \mathbb{R}$ is the utility function for appliance $a$ at time $t$ and $C_{i,a}^t: \mathbb{R} \rightarrow \mathbb{R}$ is the cost function for appliance $a$ at time $t$, and we adopt the following standard assumptions for functions $U_{i,a}^t$ and $C_{i,a}^t$:
\begin{assumption}\label{convexAssumption}
Functions $-U_{i,a}^t(q_{i,a}^D\left(t\right))$ and $C_{i,a}^t(q_{i,a}^D\left(t\right))$ are convex non-decreasing and differentiable in their respective effective domains. The gradient of the net utility function at 0 is element-wise positive, i.e. $\nabla U_{i}\left(0\right) - \nabla C_{i}\left(0\right) > 0$.
\end{assumption}

We note that, in the prosumer model (P), there is no sign restriction on $\textbf{q}_{i,a}^D$, meaning that the prosumers can either consume or produce/sell electricity (hence the name prosumer).

The set of linear inequalities $A_{i}\textbf{q}_i^D \leq \textbf{h}_{i}$ includes representations of different power consumption requirements for household appliances of prosumer $i$. Particularly, for all the appliances, we divide them into two types based on their flexibility: flexible appliances and inflexible appliances.

For flexible appliances such as ACs, washers, dryers, dishwashers, fridges, EVs/PHEVs and energy storage, whose power consumption may change depending on the price, their constraints \req{pro m2}  can be written as the following forms based on \cite{6039082}, \cite{6161320} and \cite{6840288}.
\begin{align} 
 q_{i,a}^D\left(t\right)\leq \overline{q}_{i,a}^D\left(t\right), &\ q_{i,a}^D\left(t\right) \geq \underline{q}_{i,a}^D\left(t\right) \quad \forall t,  \label{opt:PRO q}\\
\sum_{t \in \mathcal{H}_{i,a}}\alpha_{i,a}\left(t\right)q_{i,a}^D\left(t\right)\leq \overline{Q}_{i,a}^D,& \sum_{t \in \mathcal{H}_{i,a}}\alpha_{i,a}\left(t\right)q_{i,a}^D\left(t\right) \geq \underline{Q}_{i,a}^D, \label{opt:PRO e}
\end{align}
where constraint \req{opt:PRO q} represents an available range of the power consumption for the flexible appliance at time period $t$. For  AC, it means the power should between 0 and the rated power; for EV/PHEV and energy storage, it means the power should be between maximal discharging power and maximal charging power. Note that the power lower bound $\underline{q}_{i,a}^D\left(t\right)$ could be negative for EVs, PHEVs and energy storage, which reflects the definition of prosumer. Constraint \req{opt:PRO e} indicates an acceptable energy consumption range of the flexible appliance during the pre-specified periods of time $\mathcal{H}_{i,a} \subset \mathcal{H}$.\footnote{If an appliance has multiple flexible periods in $\mathcal{H}$, then the appliance has multiple constraints in form \req{opt:PRO e}. For the detailed models of appliances, interested readers are referred to \cite{6039082} and \cite{6840288}.} We call these periods in total a \textit{flexible period}, during which the appliance can adjust its power schedule as long as it can finish certain task within certain energy usage limit. For example, AC consumes enough (but not too much) energy to maintain the room temperature \cite{6039082, 6840288} in a comfortable level during the hottest time of day in summer, a washer finishes laundry by a specified time set by its user, an EV owner charges/discharges his/her car at home overnight for a trip the next day, or energy storage keeps its battery level within certain range to ensure a longer battery life. The coefficient $\alpha_{i,a}\left(t\right)$ in constraint \req{opt:PRO e} is used to capture efficiency (in a broad sense) of different types of flexible appliances. For most appliances such as washers, EVs and storage, $\alpha_{i,a}\left(t\right)$ is a power loss factor in an energy transmission process. In an ideal case, $\alpha_{i,a}\left(t\right)$ is set to constant 1, meaning there's no loss in the process. For AC, coefficient $\alpha_{i,a}\left(t\right)$ has a more special meaning, which describes time dynamic relationships between room temperature and AC power schedules.

For inflexible appliances such as lights, routers, monitor cameras, computers, TVs etc., their power consumption constraints can be written as the form of inequalities \req{opt:PRO q} (without \req{opt:PRO e}), i.e., the power consumption in each period is in a certain range.
 
In addition to representing the appliance power consumption requirements, the constraint $A_{i}\textbf{q}_i^D \leq \textbf{h}_{i}$ also describes prosumer $i$'s type of usage behavior, which can be divided into three types: \textit{simple buyers}, \textit{net buyers} and \textit{net sellers}. The simple buyers are users who only consume energy, but do not discharge energy back to grid (for their own usage or selling to others). Their DR constraints can be described by \req{opt:PRO q}-\req{opt:PRO e} with $\underline{q}_{i,a}^D\left(t\right) \geq 0$, which implies that their $q_{i,a}^{D}\left(t\right)$ is non-negative. The net buyers are users who can consume energy, store energy and discharge for their own usage but cannot sell to others. Their DR constraints can be described by  \req{opt:PRO q}-\req{opt:PRO e} and the following constraint:
\begin{gather}
\sum_{a \in\mathcal{A}_{i}} q_{i,a}^{D}\left(t\right)\geq 0,\  \forall t\label{NetSeller}.
\end{gather}
Lastly, the net sellers are those who can consume, store and sell energy to others. Their DR constraints can be covered by \req{opt:PRO q} and \req{opt:PRO e} only. To avoid trivial cases, we make the following assumption:
\begin{assumption}\label{Assumption3}
In the problem (P), the feasible set is nonempty. 
\end{assumption}
\subsection{Electric Utility Company Model} 
The electric utility company maximizes its profit (equal to revenue minus cost) by optimizing its power supply scheduling $\textbf{q}^S$. Hence, the electric utility company's problem \req{U_model} is formulated by
\begin{align}
\max_{\textbf{q}^S}  \quad
&\textbf{p}^{T}\textbf{q}^S - C_0\left(\textbf{q}^S\right),\label{U_model}\tag{U}
\end{align}
where $C_0:  \mathbb{R}^{H} \rightarrow \mathbb{R}$ is the cost function and assumed to be convex and continuously differentiable in its effective domain. To avoid trivial cases, we make the following assumption:
\begin{assumption}\label{assumptionU}
In problem (U), there exists at least one global optimal solution.
\end{assumption}

\subsection{Distribution System Operator Model}
Distribution system operator is a benevolent system planner, which maximizes the social welfare as well as keeps the supply-demand balance. The problem DSO solves, referred to as (D), is given below.
\begin{align}
\max_{ \textbf{q}^D, \textbf{q}^S}  \quad
&U\left( \textbf{q}^D\right)-C\left(\textbf{q}^D\right)-C_0\left(\textbf{q}^S\right),\label{DSO obj}\tag{D1}\\
\text{s.t.}  \quad 
&A\textbf{q}^D \leq \textbf{h},\tag{D2}\label{opt:DSO d}\\
&\textbf{q}^S=\sum_{i \in \mathcal{N}^+}\sum_{a \in \mathcal{A}_{i}}\textbf{q}_{i,a}^D,\tag{D3}\label{opt:DSO bl}
\end{align}
where the objective function in \req{DSO obj} is the social welfare. Function $U: \mathbb{R}^{H\sum_{i \in \mathcal{N}^+}|\mathcal{A}_{i}|} \rightarrow \mathbb{R}$ is prosumers' total utility function, defined by $U(\textbf{q}^D)=\sum_{i \in \mathcal{N}^+}U_i(\textbf{q}_i^D)$. Similarly, $C: \mathbb{R}^{H\sum_{i \in \mathcal{N}^+}|\mathcal{A}_{i}|} \rightarrow \mathbb{R}$ is prosumers' total cost function, defined by $C(\textbf{q}^D)=\sum_{i \in \mathcal{N}^+}C_i(\textbf{q}_i^D)$. We know that $U(\textbf{q}^D)$ and $C(\textbf{q}^D)$ are concave and convex respectively due to Assumption \ref{convexAssumption}.

The constraint $A\textbf{q}^D \leq \textbf{h}$ in \req{opt:DSO d} is an aggregation of the constraint (P2), $A_{i}\textbf{q}_i^D \leq \textbf{h}_{i}$, for all prosumers. Specifically, $A = diag\left(A_{1}, A_{2}, ..., A_{N}\right)$, $\textbf{q}^D = \left(\textbf{q}_i^D, \forall i \in \mathcal{N}^+\right)$ and $\textbf{h} = \left(\textbf{h}_i^D, \forall i \in \mathcal{N}^+\right)$.

We call the $j^{th}$ constraint for prosumer $i$ in \req{opt:DSO d} (i.e., the $j^{th}$ row of $A_{i}\textbf{q}_i^D \leq \textbf{h}_{i}$) as \textit{the $j^{th}$ general linear constraint for prosumer $i$}, which can be written as the following,

\begin{equation}
   \sum_{a \in\mathcal{A}_{i}^{\left(j\right)}} \sum_{t \in \mathcal{H}_{i,a}^{\left(j\right)}} {\alpha_{i,a}^{\left(j\right)}}\left(t\right)q_{i,a}^{D}\left(t\right)\leq [\textbf{h}_i]_j,\label{SuperConstraint}
\end{equation}
where $\mathcal{A}_{i}^{\left(j\right)} \subset \mathcal{A}_{i}$ and $\mathcal{H}_{i,a}^{\left(j\right)} \subset \mathcal{H}$. We note that constraint \req{SuperConstraint} covers all individual constraints (i.e., \req{opt:PRO q}-\req{NetSeller}) appeared in \req{opt:DSO d} and gives each individual constraint a unique label. When \eqref{SuperConstraint} represents \req{opt:PRO q} or \req{opt:PRO e}, the value $[\textbf{h}_i]_j$ specifies available ranges of power consumption for prosumer $i$'s appliances. When \eqref{SuperConstraint} represents \req{NetSeller}, $[\textbf{h}_i]_j$ restricts prosumer $i$'s net-selling amounts. In either case, $[\textbf{h}_i]_j$ reflects prosumer $i$'s capacity of resources in DR. For generality, we call each $[\textbf{h}_i]_j$ as a \textit{Resource Capacity (RC)} for prosumer $i$.

The constraint \req{opt:DSO bl} is the supply and demand balance equation. We refer to this as the \textit{market clearing condition}. We call the dual variable associated with the constraint \req{opt:DSO bl} as the \textit{market clearing price}.

\section{Efficient Market Equilibrium}\label{sec:marketequi}
We first introduce basic definitions used in this section. 
\begin{definition}[Competitive Equilibrium]
	A competitive equilibrium is defined by a demand-supply allocation and a price vector in a market such that all agents' (prosumers' and the utility company's) choices are individually optimal (for problems (P) and (U)) given the price, and the market clears.
\end{definition}
\begin{definition}[Efficient Allocation]
	An allocation is efficient if it solves optimization problem (D).
\end{definition}
The main goal of this section is to show that in the proposed $H$-period DR market, there exists a competitive equilibrium with an efficient allocation, and the competitive equilibrium set is equivalent\footnote{We say set $A$ is equivalent to set $B$, if there exists a 1-1 mapping of set A onto set B.} to the set of efficient allocations and the corresponding market clearing prices. In particular, it is described by the following theorem, whose statement partially appeared in our previous conference paper \cite{DingXiang}. Here we provide the complete version and its detailed proof.

\begin{theorem}\label{Equilibrium} In the proposed $H$-period DR market, (1) there exists an efficient demand-supply allocation $\left(\textbf{q}^D, \textbf{q}^S\right)$ and a price vector $\textbf{p}$ such that $\left(\textbf{q}^D, \textbf{q}^S, \textbf{p}\right)$ forms a competitive equilibrium;
(2) set $V_{1} = \big\{\left(\textbf{q}^D, \textbf{q}^S, \textbf{p}\right):$ $\left(\textbf{q}^D, \textbf{q}^S, \textbf{p}\right)$ is a competitive equilibrium$\big\}$, is equivalent to set $V_{2} = \big\{\left(\textbf{q}^D, \textbf{q}^S, \textbf{p}\right):$ $\left(\textbf{q}^D, \textbf{q}^S\right)$ is an efficient allocation and $\textbf{p}$ is its market clearing price$\big\}$.
\end{theorem}

\begin{proof}
(1) We prove Theorem \ref{Equilibrium} (1) by construction. First, based on Assumptions \ref{separable}-\ref{Assumption3}, we know problem (D)'s objective function is continuous and concave, and its feasible set given by \req{opt:DSO d} and \req{opt:DSO bl} is convex, nonempty and compact. By Weierstrass extreme value theorem, there exists a bounded optimal solution $\left(\textbf{q}^{D*}, \textbf{q}^{S*}\right)$. Then the optimal objective value of (D) is finite. Also due to linear forms of \req{opt:PRO q}-\req{NetSeller} and \req{opt:DSO bl}, we know problem (D) has affine constraints. Hence, by Proposition 5.3.1 on existence of primal and dual optimal solutions in \cite{Bertsekas}, the dual optimal solution for problem (D) exists. Let $\left(\boldsymbol{\lambda}^{*}, \boldsymbol{\rho}^{*}\right)$ denote a dual optimal solution pair to problem (D), where vectors $\boldsymbol{\lambda}^{*}$ and $\boldsymbol{\rho}^{*}$ are Lagrange multipliers for the constraints \req{opt:DSO d} and \req{opt:DSO bl} respectively. Then the following \textit{Karush-Kuhn-Tucker} (KKT) conditions hold:
\begin{align}
&A_{i}\textbf{q}_i^{D*}-\textbf{h}_{i} \leq 0,\ \forall i \in \mathcal{N}^{+}, \label{PF1}\tag{PF1}
\\
&\textbf{q}^{S*}=\sum_{i \in \mathcal{N}^+}\sum_{a \in \mathcal{A}_{i}}\textbf{q}_{i,a}^{D*},\label{PF2}\tag{PF2}
\\
&\boldsymbol{\lambda}_{i}^{*} \geq 0,\ \forall i \in \mathcal{N}^{+},\label{DF}\tag{DF}
\\
&diag\left \{\boldsymbol{\lambda}_{i}^{*}\left(A_{i}\textbf{q}_i^{D*}-\textbf{h}_{i}\right)^{T}\right \}=0,\ \forall i \in \mathcal{N}^{+},\label{CS}\tag{CS}
\\
&- \frac{\partial U}{\partial \textbf{q}^D}\left(\textbf{q}^{D*}\right)+\frac{\partial C}{\partial \textbf{q}^D}\left(\textbf{q}^{D*}\right)+A^{T}\boldsymbol{\lambda}^{*}+ \boldsymbol{\nu}^{*}=0,\label{FOC1}\tag{FOC1}
\\
&\frac{\partial C_0}{\partial \textbf{q}^S}\left(\textbf{q}^{S*}\right) - \boldsymbol{\rho}^{*} =0,\label{FOC2}\tag{FOC2}
\end{align}
where vector $\boldsymbol{\lambda}_{i}^{*}$ is the Lagrange multiplier that associated with prosumer $i$'s constraints in \req{opt:DSO d}, thus we have $\boldsymbol{\lambda}^{*} = \left(\boldsymbol{\lambda}_{i}^{*},\ \forall i \in \mathcal{N}^{+}\right)$. Vector $\boldsymbol{\nu}^{*}$ is given by $\boldsymbol{\nu}^{*} =\left[ \boldsymbol{\rho}^{*T}, \boldsymbol{\rho}^{*T}, ..., \boldsymbol{\rho}^{*T}\right]^T$ with $\boldsymbol{\rho}^{*T}$ repeated $\sum_{i \in \mathcal{N}^{+}}|\mathcal{A}_i|$ times.
In the above KKT conditions, constraints \req{PF1} and \req{PF2} describe primal feasibility, inequalities in \req{DF} reflect dual feasibility, equations in \req{CS} are complimentary slackness conditions, and equations \req{FOC1} and \req{FOC2} are first order optimality conditions.
When DSO sets the market price to 
\begin{align}
\textbf{p}=\boldsymbol{\rho^{*}}=\frac{\partial C_0}{\partial \textbf{q}^S}\left(\textbf{q}^{S*}\right),\nonumber
\end{align}
conditions \req{PF1}, \req{DF}, \req{CS} and \req{FOC1} guarantee $\textbf{q}_i^{D*}$ (selected from $\textbf{q}^{D*}$) and  $\boldsymbol{\lambda}_{i}^{*}$ satisfy problem (P)'s KKT condition, as shown below
\begin{align}
&A_{i}\textbf{q}_i^{D*}-\textbf{h}_{i} \leq 0,\ \boldsymbol{\lambda}_{i}^{*} \geq 0, \ diag\left \{\boldsymbol{\lambda}_{i}^{*}\left(A_{i}\textbf{q}_i^{D*}-\textbf{h}_{i}\right)^{T}\right \}=0,\nonumber
\\
&- \frac{\partial U_i}{\partial \textbf{q}_i^D}\left(\textbf{q}_i^{D*}\right)+\frac{\partial C_i}{\partial \textbf{q}_i^D}\left(\textbf{q}_i^{D*}\right)+A_{i}^{T}\boldsymbol{\lambda}_{i}^{*}+ \textbf{p}_i =0,\nonumber
\end{align}
where $\textbf{p}_i =\left[ \textbf{p}^T, \textbf{p}^T, ..., \textbf{p}^T\right]^T$ with $\textbf{p}^T$ repeated $|\mathcal{A}_{i}|$ times, and because the convex problem (P)'s strong duality holds, which is guaranteed by Assumption \ref{Assumption3} (refer to discussions of constraint qualifications in Section 5.2.3 of \cite{BoydConvex}), the quantity $\textbf{q}_i^{D*}$ is also an optimal solution to problem (P). 

Similarly, the quantity  $\textbf{q}^{S*}$ in $\left(\textbf{q}^{D*}, \textbf{q}^{S*}\right)$ is also an optimal solution to problem (U), because condition \req{FOC2} guarantees $\textbf{q}^{S*}$ satisfy problem (U)'s optimality condition, as shown below
\begin{align}
\frac{\partial C_0}{\partial \textbf{q}^S}\left(\textbf{q}^{S*}\right)- \textbf{p} =0.\nonumber
\end{align}
Hence, the pair $\left(\textbf{q}^{D*}, \textbf{q}^{S*}, \textbf{p}\right)$ form a competitive equilibrium. This completes the proof of (1).
\newline{(2)} We follow the same notation as in the first part. For any competitive equilibrium $\left(\textbf{q}^{D*}, \textbf{q}^{S*}, \textbf{p}\right) \in V_1$, let $\boldsymbol{\lambda}^{*} = \left(\boldsymbol{\lambda}_i^{*}, \forall i \in \mathcal{N^+}\right)$, where $\boldsymbol{\lambda}_i^{*}$ is the dual optimal solution to prosumer $i$'s problem (P) and let $\boldsymbol{\rho}^{*} = \textbf{p}$. Then we know $\left(\textbf{q}^{D*}, \textbf{q}^{S*}\right)$ and $\left(\boldsymbol{\lambda}^{*}, \boldsymbol{\rho}^{*}\right)$ satisfy problem (D)'s KKT conditions, since a combination of (P)'s KKT conditions, (U)'s optimality condition, the market clearing equation and the above equations together forms problem (D)'s KKT conditions. By strong duality of problem (D) due to Assumption \ref{Assumption3}, we have $\left(\textbf{q}^{D*}, \textbf{q}^{S*}\right)$ and $\left(\boldsymbol{\lambda}^{*}, \boldsymbol{\rho}^{*}\right)$ are primal and dual optimal solutions to problem (D), thus $\left(\textbf{q}^{D*}, \textbf{q}^{S*}, \textbf{p}\right) \in V_2$. Conversly, if $\left(\textbf{q}^{D*}, \textbf{q}^{S*}, \textbf{p}\right) \in V_2$, then by the same process of proof of (1), we know as long as the DSO sets price to $\textbf{p}=\frac{\partial C_0}{\partial \textbf{q}^S}\left(\textbf{q}^\textbf{S*}\right)$, $\left(\textbf{q}^{D*}, \textbf{q}^{S*}, \textbf{p}\right)$ becomes a competitive equilibrium, thus $\left(\textbf{q}^{D*}, \textbf{q}^{S*}, \textbf{p}\right) \in V_1$. This shows $V_1 = V_2$. Hence, we obtain  $V_1$ is equivalent to $V_2$.
\end{proof}
The above theorem is related to the fundamental theorems of welfare in economics \cite{FoundamentalSW}, but they are different because of the physical constraint (D2). From the theorem, we obtain an equivalence relationship between problem (D) and problems (P) and (U). This relationship makes it possible for us, by only focusing on problem (D), to analyze sensitivity of prosumers' contribution on social welfare in optimal DR with the change of the resource capacity, $[\textbf{h}_i]_j$. The analysis applies to many situations, such as the change of AC comfort zone size, allowing prosumers to be net-sellers and allowing EVs to discharge. For example, we can analyze the effect of allowing an EV to discharge on the optimal social welfare by changing the resource capacity $[\textbf{h}_i]_j$ of the EV's discharging constraint. Since all individual constraints can be described by the general linear constraint \req{SuperConstraint}, we can perform the sensitivity analysis by analyzing the dual variable (shadow price) associated with the constraint \req{SuperConstraint} of problem (D). This is what we call \textit{General Sensitivity Analysis Approach, GSAA}.

In the following sections, we study GSAA in two different net utility function settings: quadratic and general convex. In the former case, the closed-form shadow price can be derived; while in the later case, the shadow price may not be derived explicitly. Instead, we analyze properties of the bounds on the shadow price.

\section{Quadratic Settings}\label{sec:shadow quad}
In this section, we consider a setting, where all the net utility functions $\left(U_i - C_i\right), \forall i \in \mathcal{N}^+$ are quadratic. Then GSAA can provide a closed-form shadow price reflecting quantitative information about impacts of the general linear constraints on the optimal objective value, i.e., maximal social welfare.

\subsection{Model Reformulation in Quadratic Settings}
In the quadratic setting, the models of prosumers, the utility company and DSO can be reformulated into the following forms. 
\subsubsection{Prosumer}
The net utility function of prosumer $i$ using appliance $a$ at time $t$ by consuming quantity $x$ is given below,
\begin{gather}
U_{i,a}^t\left(x\right)-C_{i,a}^t\left(x\right) = \hat{a}_{i,a}\left(t\right)x^2+\hat{b}_{i,a}\left(t\right)x+\hat{c}_{i,a}\left(t\right),\label{netUtility}
\end{gather}
where the second order coefficient $\hat{a}_{i,a}\left(t\right)<0$ represents the concavity of the appliance's net utility.  The first order coefficient $\hat{b}_{i,a}\left(t\right)$ represents the appliance's initial utility increasing rate, i.e., the utility associated with consuming the first unit of energy. The constant $\hat{c}_{i,a}\left(t\right)$ represents the initial utility when time period $t$ starts. For different types of appliances, their quadratic net utility coefficients are described in Table \ref{Table Coef.}. More details on the coefficient meaning of different appliances can be found in \cite{DingXiang}.
%
\begin{table}[H]
\renewcommand{\arraystretch}{1.4}
\caption{Signs of Quadratic Net Utility Coefficients }\label{Table Coef.}
\centering
\begin{tabular}{cc|c|c|}
\cline{1-3}
\multicolumn{2}{ |c|  }{\multirow{1}{*}{Inflexible Appliances} } & $\hat{a}_{i,a}\left(t\right)<0, \hat{b}_{i,a}\left(t\right)>0, \hat{c}_{i,a}\left(t\right)=0$      \\ \cline{1-3}
\multicolumn{1}{ |c  }{\multirow{3}{*}{Flexible} } &
\multicolumn{1}{ |c| }{ACs, Washers} & $\hat{a}_{i,a}\left(t\right)<0, \hat{b}_{i,a}\left(t\right)>0 \quad \quad \quad \quad \quad \ \ $   \\ \cline{2-3}
\multicolumn{1}{ |c  }{}                        &
\multicolumn{1}{ |c| }{EVs/PHEVs} & $\hat{a}_{i,a}\left(t\right)<0, \hat{b}_{i,a}\left(t\right)>0, \hat{c}_{i,a}\left(t\right)=0$   \\ \cline{2-3}
\multicolumn{1}{ |c  }{}                        &
\multicolumn{1}{ |c| }{Energy Storage} & $\hat{a}_{i,a}\left(t\right)<0, \hat{b}_{i,a}\left(t\right)=0, \hat{c}_{i,a}\left(t\right)=0$   \\ \cline{1-3}
\end{tabular}
\end{table}

In the above table, in addition to ACs and washers, dryers, dishwashers and fridges also belong to the first category under flexible appliances.

With these quadratic coefficients, we can reformulate prosumer problem (P) into the following matrix form,
\begin{gather} 
\max_{ \textbf{q}_i^D}\frac{1}{2}\left(\textbf{q}_i^D\right)^{T}\Lambda_{i}\ \textbf{q}_i^D+\left(\textbf{b}_i-\textbf{p}_i\right)^{T}\textbf{q}_i^D+\textbf{c}_i^T\mathbb{1}_i,\text{s.t.} \ \text{\req{pro m2}}, \label{qua:i} 
\end{gather}
where matrix $\Lambda_{i}=diag\left(\Lambda_{i,1}, \Lambda_{i,2}, ...,  \Lambda_{i,|\mathcal{A}_{i}|}\right)$ with $\Lambda_{i,a} =  diag \left( 2\hat{a}_{i,a}\left(1\right), 2\hat{a}_{i,a}\left(2\right), ..., 2\hat{a}_{i,a}\left(H\right) \right)$ for $a \in {\mathcal{A}_{i}}$, vector $\textbf{b}_{i} = \left(\textbf{b}_{i,a}, \forall a \in \mathcal{A}_i\right)$ with $\textbf{b}_{i,a} = \left(\hat{b}_{i,a}\left(t\right), \forall t \in \mathcal{H}\right)$, vector $\textbf{p}_i =\left( \textbf{p},\forall a \in \mathcal{A}_i\right)$, vector $\textbf{c}_{i} = \left(\textbf{c}_{i,a}, \forall a \in \mathcal{A}_i\right)$ with $\textbf{c}_{i,a} = \left(\hat{c}_{i,a}\left(t\right), \forall t \in \mathcal{H}\right)$, and vector $\mathbb{1}_i = \left(\left(1,  \forall t \in \mathcal{H}\right), \forall a \in \mathcal{A}_i\right)$.
\subsubsection{Electric Utility Company}
For simplicity, we assume the utility company's production cost is linear and can be time-varying, i.e., for $\textbf{q}^S \in \mathbb{R}^H$,
\begin{align}
C_0\left(\textbf{q}^S\right)=\textbf{b}_{0}^{T}\textbf{q}^S,\quad \text{where $\textbf{b}_{0} \in \mathbb{R}_+^H$.} \label{qua:dc} 
\end{align}

\subsubsection{DSO}\label{DSOsec}
Under the quadratic framework and by using equation \req{opt:DSO bl} to substitute decision variable $\textbf{q}^S$ with $\textbf{q}^D$, the DSO problem (D) can be reformulated into the following quadratic programming problem,
\begin{align}
\max_{ \textbf{q}^D} \quad
\frac{1}{2}\left(\textbf{q}^D\right)^{T}\Lambda\ \textbf{q}^D+\bm{\bar{b}}^{T}\textbf{q}^D+\textbf{g}^T\mathbb{1} \quad \text{s.t.} \ \text{\req{opt:DSO d}},\label{DSO q}
\end{align}
where matrix $\Lambda = diag\left(\Lambda_{1}, \Lambda_{2}, ..., \Lambda_{N}\right)$, vectors $\bm{\bar{b}} = \left( \left( \textbf{b}_{i,a} - \textbf{b}_{0}, \forall a \in \mathcal{A}_i\right), \forall i \in \mathcal{N}^+\right)$, $\textbf{g} = \left(\textbf{c}_{i}, \forall i \in \mathcal{N}^+\right)$ and $\mathbb{1} = \left(\mathbb{1}_i, \forall i \in \mathcal{N}^+\right)$. Matrix $\Lambda$ is a full rank diagonal $H\sum_{i \in \mathcal{N}^{+}}|\mathcal{A}_i|$ by $H\sum_{i \in \mathcal{N}^{+}}|\mathcal{A}_i|$ matrix.

\subsection{Shadow Price in Quadratic Settings}\label{ShadowQuadSettings}
In this section, we derive a closed-form expression of the shadow price, i.e., the Lagrange multiplier/the dual variable associated to the $j^{th}$ general linear constraint for prosumer $i$, \req{SuperConstraint}, in the reformulated DSO problem \req{DSO q}. Let $\boldsymbol{\lambda}$ denote the Lagrange multiplier vector to constraint set \req{opt:DSO d}. First, we show an important property of the problem \req{DSO q} used in deriving the shadow price later.
\begin{property}\label{Qual_De_Equal}
	The reformulated DSO problem \req{DSO q} is equivalent to the following $N$ problems: $\forall i \in \mathcal{N}^+$,
	\begin{align} 
	\max_{ \textbf{q}_i^D}\frac{1}{2}\left(\textbf{q}_i^D\right)^{T}\Lambda_{i}\ \textbf{q}_i^D+\bm{\bar{b}}_i^{T}\textbf{q}_i^D+\textbf{c}_i^T\mathbb{1}_i,\ \text{s.t.} \ \text{\req{pro m2}},\label{De_quad} 
	\end{align}
	where $\bm{\bar{b}}_i = \left( \textbf{b}_{i,a} - \textbf{b}_{0}, \forall a \in \mathcal{A}_i\right)$. By equivalent, we mean they have the same optimal primal solutions and the same optimal dual solutions. (Problem \req{De_quad} is the same as the reformulated prosumer problems \req{qua:i} with $\textbf{p}_i = \textbf{b}_{0}$, for all $i \in \mathcal{N}^+$.)
\end{property}
\begin{proof}
	Notice that the objective function of problem \req{DSO q} is equal to $\sum_{i  \in\mathcal{N}^+}\left[\frac{1}{2}\left(\textbf{q}_i^D\right)^{T}\Lambda_{i}\ \textbf{q}_i^D+\bm{\bar{b}}_i^{T}\textbf{q}_i^D+\textbf{c}_i^T\mathbb{1}_i\right]$, where $\bm{\bar{b}}_i = \left( \textbf{b}_{i,a} - \textbf{b}_{0}, \forall a \in \mathcal{A}_i\right)$, and in \req{opt:DSO d}, the constraints corresponding to one prosumer is independent from the constraints corresponding to the other prosumers. Hence, the KKT conditions for the original problem \req{DSO q} are the same as the combination of KKT conditions of all decoupled problems \req{De_quad} (One can refer to the KKT conditions used in Theorem \ref{Equilibrium}'s proof). Hence, by strong duality of both problems \req{DSO q} and \req{De_quad} for all $i \in \mathcal{N}^+$, their optimal primal and dual solutions are the same to each other.
\end{proof}
The above property tells us, deriving the shadow price for the DSO problem \req{DSO q} is equivalent to deriving the shadow price for the decoupled problems \req{De_quad}. Furthermore, the shadow price of the decoupled problems \req{De_quad} can be directly applied to sensitivity analysis to the original problem \req{DSO q}.

Now, we can first focus on the decoupled problem \req{De_quad} for prosumer $i$. Let $\boldsymbol{\lambda}_i$ denote the optimal dual variable associated with constraint \req{pro m2}. Since the matrix $\Lambda_i$ is diagonal and of full rank, by first order condition we can write the optimal solution $\textbf{q}_i^{D*}$ for problem \req{De_quad} as 
\begin{align}
\textbf{q}_i^{D*}=-\Lambda_i^{-1}\left(\bm{\bar{b}}_i^{T}-A_i^{T}\boldsymbol{\lambda}_i\right). \label{OptWithDual}
\end{align}
Substituting \req{OptWithDual} to \req{De_quad} and by strong duality, the quadratic programming problem \req{De_quad} is equivalent to its dual problem as below,
\begin{align}
\max_{ \boldsymbol{\lambda}_i \geq 0} \quad
&  \frac{1}{2}\boldsymbol{\lambda}_i^{T}A_i\Lambda_i^{-1}A_i^{T}\boldsymbol{\lambda}_i-\boldsymbol{\lambda}_i^{T}\left(A_i\Lambda_i^{-1}\bm{\bar{b}}_i^{T}+\textbf{h}_i\right),\label{qua dual}
\end{align}
where we ignore the constant term $\textbf{c}_i^T\mathbb{1}_i$, since it does not affect the optimal solution.

For the above dual problem \req{qua dual}, to the best of our knowledge, there is no explicit closed-form solution. However, for this paper, we are interested in estimating the effect of changing one of prosumer $i$'s resource capacities, $[\textbf{h}_i]_j$, at a time on the social welfare. Hence, we only need to focus on one element of Lagrange multiplier vector $\boldsymbol{\lambda}_i$.

We do so by focusing on the targeted constraint i.e., the $j^{th}$ general linear constraint for prosumer $i$, \req{SuperConstraint}, and assume this constraint is tight at optimality while prosumer $i$'s other constraints are not (Note, there is no assumption about the other prosumers' constraints). Then derive the shadow price $\left[\boldsymbol{\lambda}_i^*\right]_j$ associated with the constraint \req{SuperConstraint}. 
Due to complimentary slackness, we know for prosumer $i$, $\left[\boldsymbol{\lambda}_i^*\right]_l=0, \forall l \neq j$. By matrix calculations, we can obtain the shadow price $\left[\boldsymbol{\lambda}_i^*\right]_j$ of the target constraint \req{SuperConstraint} as below,
\begin{align}
\left[\boldsymbol{\lambda}_i^*\right]_j=\left[\frac{\left[A_i \Lambda_i^{-1}\bm{\bar{b}}_i + \textbf{h}_i\right]_j}{\left[A_i \Lambda_i^{-1}A_i^T\right]_{jj}}\right]^{+}. \label{MatrixShadow}
\end{align}
By the definitions of $A$ and $\Lambda$, we have
\begin{align}
\left[\boldsymbol{\lambda}_i^*\right]_j=\left[\frac{\sum_{a \in\mathcal{A}_{i}^{\left(j\right)}}\sum_{t \in \mathcal{H}_{i,a}^{\left(j\right)}} \frac{\alpha_{i,a}^{\left(j\right)}(t)\left[\hat{b}_{i,a}\left(t\right) - b_{0}\left(t\right)\right]}{2\hat{a}_{i,a}\left(t\right)} + [\textbf{h}_i]_j}{\sum_{a \in\mathcal{A}_{i}^{\left(j\right)}}\sum_{t \in \mathcal{H}_{i,a}^{\left(j\right)}} \frac{\left[\alpha_{i,a}^{\left(j\right)}(t)\right]^2}{2\hat{a}_{i,a}\left(t\right)}}\right]^{+}. \label{SuperShadow}
\end{align}
Based on property \ref{Qual_De_Equal}, we know the above closed-form shadow price \req{SuperShadow} reflects the rate of change in social welfare (objective function value of (D)) associated with change of the resource capacity $[\textbf{h}_i]_j$. More specifically, the amount of social welfare improvement associated with increasing the resource capacity by up to $K$ units, i.e., relaxing the general linear constraint \req{SuperConstraint} to 
$\sum_{a \in\mathcal{A}_{i}^{\left(j\right)}} \sum_{t \in \mathcal{H}_{i,a}^{\left(j\right)}} {\alpha_{i,a}^{\left(j\right)}\left(t\right)}q_{i,a}^{D}\left(t\right)\leq [\textbf{h}_i]_j+K $, is upper bounded by $K\left[\boldsymbol{\lambda}_i^*\right]_j$. Hence, the closed-form shadow price \req{SuperShadow} can be used as a measuring tool for GSAA to quickly estimate the potential of prosumer $i$'s contribution on social welfare with change of the resource capacity in DR. When one needs to encourage/select customers to participate in DR or make investment on network construction or upgrade to enable certain DR functions, priority should be given to the prosumer with a larger potential contribution i.e., a larger value of $K\lambda_{AC}^{*}$ (especially when there is a limited budget for the campaign, infrastructure upgrade or management).

\subsection{GSAA Applications in Quadratic Settings}\label{GSAA-Q app}
To show how GSAA can be used to identify contribution potentials of different prosumers on social welfare with change of their resource capacities in DR, we provide three application examples here: enlarging AC comfort zone, allowing net-selling behavior and allowing EVs to discharge.
\subsubsection{GSAA for Enlarging AC Comfort Zone}\label{GSAA_AC_Q}
For prosumer $i$ using an AC (labeled by $AC$) at time period $t$, the AC comfort zone constraint (refer to \cite{6039082} and \cite{6840288}) can be represented by  
\begin{gather}
\underline{Q}_{i,{AC}}^D\leq \sum_{\tau = 1}^{t} \alpha_{i,AC}\left(\tau\right)q_{i,AC}^{D}\left(\tau\right)\leq \overline{Q}_{i,AC}^D,\label{ACoriginal}
\end{gather}
where, $\alpha_{i,AC}\left(\tau\right) < 0$. Values $\underline{Q}_{i,{AC}}^D$ and $\overline{Q}_{i,AC}^D$ reflect prosumer $i$'s equivalent lowest and highest comfort temperatures  respectively. Due to symmetry, here we only take the right inequality of \req{ACoriginal} as an example. It can be reformulated as the $j$th general linear constraint for prosumer $i$, of form \eqref{SuperConstraint},
where $\mathcal{A}_{i}^{\left(j\right)} = \{AC\}$, $\mathcal{H}_{i,a}^{\left(j\right)} = \{1, 2, ..., t\}$, $\alpha_{i,a}^{\left(j\right)}\left(t\right) = \alpha_{i,a}\left(t\right)$ and $[\textbf{h}_i]_j = \overline{Q}_{i,AC}^D$. We use $\lambda_{AC}^*$ to denote the optimal dual variable, i.e., the shadow price, corresponding to this constraint. By using GSAA, i.e., substituting specific parameters of the constraint and the net utility function \req{netUtility} into \req{SuperShadow}, we can obtain the shadow price for relaxing AC comfort zone's upper limit as
\begin{align}
\lambda_{AC}^{*}=\left[\frac{\sum_{\tau = 1}^{t}\frac{  \alpha_{i,AC}\left(\tau\right)\left(\hat{b}_{i,AC}\left(\tau\right) - b_{0}\left(\tau\right)\right)}{\hat{a}_{i,AC}\left(\tau\right)} + 2\overline{Q}_{i,AC}^D}{\sum_{\tau = 1}^{t}  \frac{\alpha_{i,AC}^2\left(\tau\right)}{\hat{a}_{i,AC}\left(\tau\right)}} \right]^{+}. \label{sen1}
\end{align}
Then $K\lambda_{AC}^{*}$ shows the potential of prosumer $i$'s contribution on social welfare by enlarging his/her AC comfort zone limit $\overline{Q}_{i,AC}^D$ by $K$ units. From the form of 
\req{sen1}, we can observe that for two prosumers with similar utility functions and similar AC parameters $(\alpha_{i,AC}\left(\tau\right))$, the one who accepts a higher temperature, reflected as larger $\overline{Q}_{i,AC}^D$, has a lower contribution potential on social welfare. This result suggests that we should target those prefer cold AC setpoints first.

\subsubsection{\textit{GSAA for Allowing Net Selling} }\label{GSAA_NS_Q}
For prosumer $i$, who is also a net buyer, we can use GSAA to estimate his/her contribution on social welfare by allowing him/her to be a net seller in DR, i.e., relaxing the net buying constraint \req{NetSeller}. This constraint can be reformulated as the $j$th general linear constraint for prosumer $i$, of form \eqref{SuperConstraint}, where $\mathcal{A}_{i}^{\left(j\right)} = \mathcal{A}_{i}$, $\mathcal{H}_{i,a}^{\left(j\right)} = \{t\}$, $\alpha_{i,a}^{\left(j\right)} = -1$ and $[\textbf{h}_i]_j = 0$. We use $\lambda_{NS}^*$ to denote the optimal dual variable corresponding to this constraint. By using GSAA, i.e., substituting specific parameters of the constraint and the net utility function \req{netUtility} into \req{SuperShadow}, we can obtain the shadow price for allowing a prosumer to be a net seller as
\begin{align}
\lambda_{NS}^{*}=\left[\frac{\sum_{a \in\mathcal{A}_{i}} \frac{b_{0}\left(t\right)-\hat{b}_{i,a}\left(t\right)}{\hat{a}_{i,a}\left(t\right)}}{\sum_{a \in\mathcal{A}_{i}} \frac{1}{\hat{a}_{i,a}\left(t\right)}}\right]^{+}. \label{sen2}
\end{align}
Then $K\lambda_{NS}^{*}$ shows the potential of the net buyer $i$'s contribution on social welfare in DR by allowing him/her to net sell energy up to $K$ units. From the form of \req{sen2}, we can find that the effect of allowing net selling on social welfare improvement is not due to a single appliance. Instead, it is affected by a combination of all appliances, which implies the high complexity of accurate sensitivity analysis. If there is one appliance, whose net utility function is much flatter (very price-responsive) than the rest, i.e., there exists an $a_1$ with $\hat{a}_{i,a_{1}}\left(t\right) >> \hat{a}_{i,a_{2}}\left(t\right), ..., \hat{a}_{i,a_{|\mathcal{A}_{i}|}}\left(t\right)$, and if all appliance initial utility increasing rates $\hat{b}_{i,a_{1}}\left(t\right), ..., \hat{b}_{i,a_{|\mathcal{A}_{i}|}}\left(t\right)$ are similar, then $\lambda_{NS}^{*}=\left[b_{0}(t)-\hat{b}_{i,a_{1}}(t)\right]^{+}$. Hence, the shadow price is dominated by this most price-responsive appliance.

\subsubsection{GSAA for Allowing EVs to Discharge} \label{GSAA_EV_Q}
By using GSAA, we can also estimate the effect of allowing prosumer $i$'s EV (labeled by $EV$) to discharge at time period $t$, i.e., relaxing the constraint $q_{i,a}^D\left(t\right) \geq 0$, on the social welfare improvement. Similarly, the above constraint can be reformulated as the $j$th general linear constraint for prosumer $i$, of form \eqref{SuperConstraint}, where $\mathcal{A}_{i}^{\left(j\right)} = \{EV\}$, $\mathcal{H}_{i,a}^{\left(j\right)} = \{t\}$, $\alpha_{i,a}^{\left(j\right)} = -1$ and $[\textbf{h}_i]_j = 0$. We use $\lambda_{EV}^*$ to denote the optimal dual variable associated to this constraint.  Substituting specific parameters of the constraint and the net utility function \req{netUtility} into \req{SuperShadow}, we can obtain the shadow price for allowing EVs to discharge is
\begin{align}
\lambda_{EV}^{*}=\left[b_{0}\left(t\right)-\hat{b}_{i,EV}\left(t\right)\right]^{+}. \label{sen3}
\end{align}
Similarly, if we allow the EV to discharge up to $K$ units of energy, then the potential of prosumer $i$'s contribution on social welfare is $K\lambda_{EV}^{*}$. We can see that the above shadow price is monotonically increasing in the current utility company's production cost and monotonically decreasing in $\hat{b}_{i,EV}$. Hence those with low values of $\hat{b}_{i,EV}$, who value EV charging less, should be given EV discharging capacity first.
\section{General Convex Settings}\label{sec:shadow conv}
In this section, we study GSAA under a more general framework, where we do not have an explicit form of a prosumer's net utility function except for some \textit{nice convexity}, by which we mean the prosumer's negative net utility function is strongly convex and has Lipschitz continuous gradient. It's worth mentioning that the quadratic settings discussed in the previous section is just a special case of general convex settings here. Under the general convex settings, even though we cannot get a closed-form shadow price, we are able to derive upper and lower bounds of shadow prices based on the quadratic settings before. These bounds, as indirect information, can still provide us insights on prosumers' marginal contributions on social welfare as they change resources capacity in DR.

\subsection{Preliminary}
We first introduce basic definitions and equivalent conditions for strong convexity and Lipschitz continuous gradient\cite{BoydConvex, Vandenberghe}, which are used in theorems appeared latter.
\begin{definition}[Strong Convexity]
A function $f:\mathbb{R}^n\rightarrow\mathbb{R}$ is \textit{$\mu$}-strongly convex (or has \textit{$\mu$-strong convexity}) for $\mu>0$, if the function $f(\textbf{x})-\frac{\mu}{2}\|\textbf{x}\|^{2}$ is convex for all $\textbf{x} \in \mathbb{R}^{n}$.
\end{definition}
\begin{lemma}\label{lemma1}
A differentiable function $f: \mathbb{R}^n\rightarrow \mathbb{R}$ is \textit{$\mu$}-strongly convex if and only if $\ \forall \textbf{x}, \textbf{y} \in \mathbb{R}^{n}$,
\begin{align} \label{opt:Supply}
&f(\textbf{x})\geq f(\textbf{y})+\nabla f(\textbf{y})^{T}(\textbf{x}-\textbf{y})+\frac{\mu}{2}\|\textbf{x}-\textbf{y}\|^2.
\end{align}
\end{lemma}
\begin{lemma} 
If a function $f: \mathbb{R}^n\rightarrow \mathbb{R}$ is $\mu$-strongly convex, then $f$ is strictly convex. \cite{Ahmadi}
\end{lemma}
\begin{definition}[Lipschitz Continuous Gradient] \label{Def2}
A differentiable continuous function $f: \mathbb{R}^n\rightarrow \mathbb{R}$ has \textit{$L$-Lipschitz continuous gradient} for $L>0$, if $\|\nabla f(\textbf{x})-\nabla f(\textbf{y})\| \leq L\|\textbf{x}-\textbf{y}\|$ for all $\textbf{x}, \textbf{y} \in \mathbb{R}^n$.
\end{definition}
\begin{lemma}\label{lemma2}
A convex function $f: \mathbb{R}^n\rightarrow \mathbb{R}$ has \textit{$L$-Lipschitz continuous gradient} if and only if $\ \forall \textbf{x}, \textbf{y} \in R^{n}$,
\begin{align}
f(\textbf{x})\leq f(\textbf{y})+\nabla f(\textbf{y})^{T}(\textbf{x}-\textbf{y})+\frac{L}{2}\|\textbf{x}-\textbf{y}\|^2
\end{align}
\end{lemma}

\subsection{Model Reformulation in General Convex Settings}
For notational convenience, we let $F_i(\textbf{q}_i^D)=- \left[U_i\left(\textbf{q}_i^D\right)\right.$ $\left.-C_i\left(\textbf{q}_i^D\right)\right]$. In the general convex framework, all prosumers' net utility functions have nice convexity. Specifically, we adopt the following assumption: 
\begin{assumption}\label{AssumpMuStrong}
$F_i(\textbf{q}_i^D)$ is differentiable, $\mu_i$-strongly convex and has $L_i$-Lipschitz continuous gradient, $\forall i \in\mathcal{N}_{+}$.
\end{assumption}  

With notation $F_i(\textbf{q}_i^D)$, we can write DSO problem (D) as
\begin{align}
\min_{ \textbf{q}^D, \textbf{q}^S}  \quad
\sum_{i \in\mathcal{N}_{+}}F_i(\textbf{q}_i^{D})+C_0\left(\textbf{q}^S\right) \quad \text{s.t.}\  \text{\req{opt:DSO d} and \req{opt:DSO bl}.}\nonumber 
\end{align}

For simplicity, we still adopt the assumption that the utility company's cost function $C_0\left(\textbf{q}^S\right)$ has linear form as in \eqref{qua:dc}. By substituting \eqref{qua:dc} and \eqref{opt:DSO bl} into the above DSO problem (D), we can simplify it to
\begin{align}
\min_{ \textbf{q}^D}  \quad
&\sum_{i \in\mathcal{N}_{+}}F_i(\textbf{q}_i^{D})+\sum_{i \in\mathcal{N}_{+}}\sum_{a \in\mathcal{A}_{i}}\textbf{b}_0^{T}\textbf{q}_{i,a}^D\label{DSOnewObj2}  \quad \text{s.t.}\ \text{\req{opt:DSO d}.}
\end{align}
Let $F\left(\textbf{q}^{D}\right)$ denote the objective function in \req{DSOnewObj2},
then by Assumption \ref{AssumpMuStrong}, Lemma \ref{lemma1} and \ref{lemma2} for each $F_i(\textbf{q}_i^D)$ and summing over all $i \in\mathcal{N}_{+}$, we have:
$\ \forall \textbf{q}^{D},  \underbar{$\textbf{r}$},  \overline{\textbf{r}} \in \mathbb{R}^{H\sum_{i \in\mathcal{N}_{+}}|\mathcal{A}_{i}|}$,
\begin{align}
&\underbar{$F$}(\underbar{$\textbf{r}$},\textbf{q}^{D}) \leq F(\textbf{q}^{D}) \leq \overline{F}(\overline{\textbf{r}}, \textbf{q}^{D}),\label{arbitary}
\end{align}
where
\begin{align*}
&\underbar{$F$}(\underbar{$\textbf{r}$},\textbf{q}^{D})=\sum_{i \in\mathcal{N}_{+}}\underbar{$F$}_i(\underbar{$\textbf{r}$}_i, \textbf{q}_i^{D})+\sum_{i \in\mathcal{N}_{+}}\sum_{a \in\mathcal{A}_{i}}\textbf{b}_0^{T}\textbf{q}_{i,a}^D,\  \underbar{$\textbf{r}$}=\left(\underbar{$\textbf{r}$}_i, \forall i \right),\\
&\overline{F}(\overline{\textbf{r}}, \textbf{q}^{D})=\sum_{i \in\mathcal{N}_{+}}\overline{F}_i(\overline{\textbf{r}}_i, \textbf{q}_i^{D})+\sum_{i \in\mathcal{N}_{+}}\sum_{a \in\mathcal{A}_{i}}\textbf{b}_0^{T}\textbf{q}_{i,a}^D,\  \overline{\textbf{r}}=\left(\overline{\textbf{r}}_i, \forall i \right),\\
&\underbar{$F$}_i(\underbar{$\textbf{r}$}_i,\textbf{q}_i^{D})=F_i(\underbar{$\textbf{r}$}_i)+\nabla F_i(\underbar{$\textbf{r}$}_i)^{T}(\textbf{q}_i^{D}- \underbar{$\textbf{r}$}_i)+\frac{\mu_i}{2}\|\textbf{q}_i^{D}- \underbar{$\textbf{r}$}_i\|^{2},\\
&\overline{F}_i(\overline{\textbf{r}}_i,\textbf{q}_i^{D})=F_i(\overline{\textbf{r}}_i)+\nabla F_i(\overline{\textbf{r}}_i)^{T}(\textbf{q}_i^{D}- \overline{\textbf{r}}_i)+\frac{L_i}{2}\|\textbf{q}_i^{D}- \overline{\textbf{r}}_i\|^{2}.
\end{align*}

\subsection{Shadow Price in General Convex Settings}
In this section we derive closed-form upper and lower bounds of shadow price, i.e., Lagrange multiplier associated with any general linear constraint \req{SuperConstraint}, from the following three problems.
\begin{align}
\min_{\textbf{q}^{D}}  \quad
\underbar{$F$}(\underbar{$\textbf{r}$},\textbf{q}^{D})
\qquad\text{s.t.}\ 
\text{\req{opt:DSO d}},\tag{I}\label{F1}\\
\min_{\textbf{q}^{D}}  \quad
F\left(\textbf{q}^{D}\right)
\qquad\text{s.t.}\ 
\text{\req{opt:DSO d}},\tag{II}\label{F}\\
\min_{\textbf{q}^{D}}  \quad
\overline{F}(\overline{\textbf{r}}, \textbf{q}^{D})
\qquad\text{s.t.}\ 
\text{\req{opt:DSO d}},\tag{III}\label{F3}
\end{align}
where the problem (II) is the problem \req{DSOnewObj2} with a simplified form. Some basic properties of the problems (I), (II) and (III) are given below. 
\begin{property}\label{Unique123}
For any given $\underbar{$\textbf{r}$}$ and $\overline{\textbf{r}}$ ($\|\underbar{$\textbf{r}$}\|, \|\overline{\textbf{r}}\| <\infty$), each problem of (I), (II) and (III) has unique primal and dual optimal solutions.
\end{property}
\begin{proof}
Given any $\underbar{$\textbf{r}$}$ and $\overline{\textbf{r}}$ ($\|\underbar{$\textbf{r}$}\|, \|\overline{\textbf{r}}\| <\infty$), by a similar process of proof of Theorem \ref{Equilibrium}, the dual optimal solutions for (I), (II) and (III) exist. Next, since $F_i\left(\textbf{q}_{i}^{D}\right)$ is $\mu_i$-strongly convex, by Lemma \ref{lemma1}, we know $F_i\left(\textbf{q}_{i}^{D}\right)$ is strictly convex. Since $F\left(\textbf{q}^{D}\right)$ is the summation of strictly convex functions and convex functions, it is strictly convex too. Similarly, because $\underbar{$F$}(\underbar{$\textbf{r}$},\textbf{q}^{D})$ and $\overline{F}(\overline{\textbf{r}}, \textbf{q}^{D})$ are summations of quadratic functions (strictly convex) and a linear function (convex), they are also stricly convex for any $\underbar{$\textbf{r}$}^{*}$ and $\overline{\textbf{r}}^{*}$. Hence, due to the strict convexity, by Section 5.5.5 of \cite{BoydConvex}, we know both primal and dual optimal solutions are unique for each problem.
\end{proof}
\begin{property}
The optimal objective values $\underbar{$F$}^*$, $F^*$, $\overline{F}^*$ of the problems (I), (II), (III) satisfy $\underbar{$F$}^* \leq F^* \leq \overline{F}^*$ for any given $\underbar{$\textbf{r}$}$ and $\overline{\textbf{r}}$.
\end{property}
\begin{proof}
	This holds due to \eqref{arbitary}.
\end{proof}
With the above property, we can denote these unique primal and dual optimal solution pairs for problems (I), (II) and (III) by $\left(\underline{\textbf{q}}^{D*}, \underline{\boldsymbol{\lambda}}^*\right)$, $\left(\textbf{q}^{D*}, \boldsymbol{\lambda}^*\right)$ and $\left(\overline{\textbf{q}}^{D*}, \overline{\boldsymbol{\lambda}}^{*}\right)$  respectively. A relationship between the optimal solutions of these three problems are described in the following theorem.
\begin{property}\label{EqualSolutions}
There exist vectors $\underbar{$\textbf{r}$}^{*}$ and $\overline{\textbf{r}}^{*}$, such that optimization problems (I), (II) and (III) have the same optimal solution i.e., $\underline{\textbf{q}}^{D*}=\textbf{q}^{D*}=\overline{\textbf{q}}^{D*}$.
\end{property}
\begin{proof}
We first show by construction that there exists a vector $\underline{\textbf{r}}^{*}$ such that problem (I) and (II) have the same optimal solution i.e., $\underline{\textbf{q}}^{D*}=\textbf{q}^{D*}$. 
Firstly, we solve problem (I) with $\underline{\textbf{r}}=\textbf{q}^{D*}$. By definition of $\underbar{$F$}(\underbar{$\textbf{r}$},\textbf{q}^{D})$, the objective function of problem (I) becomes
\begin{align}
&\underline{F}\left(\textbf{q}^{D*},\textbf{q}^{D}\right)= \sum_{i \in\mathcal{N}_{+}}\frac{\mu_i}{2}\|\textbf{q}_i^{D}- \textbf{q}_i^{D*}\|^{2} +\sum_{i \in\mathcal{N}_{+}}\sum_{a \in\mathcal{A}_{i}}\textbf{b}_0^{T}\textbf{q}_{i,a}^D\nonumber \\
&+ \sum_{i \in\mathcal{N}_{+}}F_i(\textbf{q}_i^{D*})+\sum_{i \in\mathcal{N}_{+}}\nabla F_i(\textbf{q}_i^{D*})^{T}(\textbf{q}_i^{D}-\textbf{q}_i^{D*}).
\end{align}
Organizing the above equation in terms of the order of decision variable $\textbf{q}^{D}$, we have
\begin{align}
&\underline{F}\left(\textbf{q}^{D*},\textbf{q}^{D}\right)=\frac{1}{2}\sum_{i \in\mathcal{N}_{+}}\mu_i \left(\textbf{q}_i^{D}\right)^{T}\textbf{q}_i^{D} +\sum_{i \in\mathcal{N}_{+}}\sum_{a \in\mathcal{A}_{i}}\textbf{b}_0^{T}\textbf{q}_{i,a}^D\nonumber 
\\
&-\sum_{i \in\mathcal{N}_{+}}\left[\mu_i \textbf{q}_i^{D*} - \nabla F_i\left( \textbf{q}_i^{D*}\right)\right]^{T}\textbf{q}_i^{D} + \sum_{i \in\mathcal{N}_{+}} \left[F_i(\textbf{q}_i^{D*})\right. \nonumber
\\
& \left.- \nabla F_i\left( \textbf{q}_i^{D*}\right)^{T}\textbf{q}_i^{D*} + \frac{\mu_i}{2}\|\textbf{q}_i^{D*}\|^2\right],
\end{align}
where the last summation term is a constant.
Hence, KKT conditions for the problem (I) with $\underline{\textbf{r}}=\textbf{q}^{D*}$ are as below,
\begin{align}
&\underline{\boldsymbol{\lambda}}^{*} \geq 0, \nonumber \tag{I.1}\label{I.1}
\\
&A\underline{\textbf{q}}^{D*}-\textbf{h} \leq 0, \nonumber \tag{I.2}\label{I.2}
\\
&diag\left \{\underline{\boldsymbol{\lambda}}^{*}\left(A\underline{\textbf{q}}^{D*}-\textbf{h}\right)^{T}\right \}=0,\nonumber \tag{I.3}\label{I.3}
\\
&M \underline{\textbf{q}}^{D*}-\left[M \textbf{q}^{D*}-\nabla F\left( \textbf{q}^{D*}\right)\right] +A^{T}\underline{\boldsymbol{\lambda}}^{*}=0,\nonumber \tag{I.4}\label{I.4}
\end{align}
where $M = diag\left(\mu_1 \textbf{I}_1, \mu_2 \textbf{I}_2, ..., \mu_N \textbf{I}_N\right)$. $\textbf{I}_i$ is an identity matrix with length $H |\mathcal{A}_{i}|$.
If we replace $(\underline{\textbf{q}}^{D*}, \underline{\boldsymbol{\lambda}}^{*})$ by $(\textbf{q}^{D*}, \boldsymbol{\lambda}^{*})$ in the above KKT conditions for the problem (I), then we can obtain the following expressions
\begin{align}
&\boldsymbol{\lambda}^{*} \geq 0, \nonumber \tag{II.1}\label{II.1}\\
&A\textbf{q}^{D*}-\textbf{h} \leq 0, \nonumber \tag{II.2}\label{II.2}\\
&diag\left \{\boldsymbol{\lambda}^{*}\left(A\textbf{q}^{D*}-\textbf{h}\right)^{T}\right \}=0,\nonumber \tag{II.3}\label{II.3}\\
&\nabla F( \textbf{q}^{D*})+A^{T}\boldsymbol{\lambda}^{*}=0.\nonumber \tag{II.4} \label{II.4}
\end{align} 
These expressions hold due to optimality of $(\textbf{q}^{D*}, \boldsymbol{\lambda}^{*})$ as they are exactly the same as KKT conditions for problem (II). Hence, $\textbf{q}^{D*}$ also satisfies \req{I.1}-\req{I.4} and is an optimal solution of the problem (I). With Theorem \ref{Unique123}, we know $\underline{\textbf{q}}^{D*}=\textbf{q}^{D*}$.\\
%
By a similar process, we can prove for $\overline{\textbf{r}}^{*} = \textbf{q}^{D*}$, (II) and (III) have the same optimal solution, i.e., $\textbf{q}^{D*}=\overline{\textbf{q}}^{D*}$.
\footnote{It's worth mentioning that the proof only shows existence, but not provide information about uniqueness of such existence. }
\end{proof}
\begin{property}\label{Conv_De_Property}
	The problems (I), (II) and (III) are equivalent to the following three decoupled problems respectively, i.e., each of problems (I), (II) and (III) has the same primal and dual optimal solutions to its decoupled problem.
	\begin{align}
		\min_{\textbf{q}_i^{D}}  \quad
		\underbar{$G$}_i(\underbar{$\textbf{r}$}_i,\textbf{q}_i^{D})
		\qquad\text{s.t.}\ 
		\text{\req{pro m2}},\quad \forall i \in\mathcal{N}_{+},  \tag{De-I}\label{D_F1}\\
		\min_{\textbf{q}_i^{D}}  \quad
		G_i\left(\textbf{q}_i^{D}\right)
		\qquad\text{s.t.}\ 
		\text{\req{pro m2}},\quad \forall i \in\mathcal{N}_{+},  \tag{De-II}\label{D_F}\\
		\min_{\textbf{q}_i^{D}}  \quad
		\overline{G}_i(\overline{\textbf{r}}_i, \textbf{q}_i^{D})
		\qquad\text{s.t.}\ 
		\text{\req{pro m2}},\quad \forall i \in\mathcal{N}_{+},  \tag{De-III}\label{D_F3}
	\end{align}
	where $\underbar{$G$}_i(\underbar{$\textbf{r}$}_i,\textbf{q}_i^{D}) = \underbar{$F$}_i(\underbar{$\textbf{r}$}_i, \textbf{q}_i^{D})+\sum_{a \in\mathcal{A}_{i}}\textbf{b}_0^{T}\textbf{q}_{i,a}^D$, $G_i(\textbf{q}_i^{D}) = F_i(\textbf{q}_i^{D})+\sum_{a \in\mathcal{A}_{i}}\textbf{b}_0^{T}\textbf{q}_{i,a}^D$ and $\ \overline{G}_i(\overline{\textbf{r}}_i, \textbf{q}_i^{D}) = \overline{F}_i(\overline{\textbf{r}}_i, \textbf{q}_i^{D})+\sum_{a \in\mathcal{A}_{i}}\textbf{b}_0^{T}\textbf{q}_{i,a}^D$.
\end{property}
\begin{proof}
	Due to \eqref{arbitary}, we know the problems (I), (II) and (III) are summations of problems \req{D_F1}, \req{D_F} and \req{D_F3} over $i \in\mathcal{N}_{+}$ respectively. And because different prosumers' constraints are separable, following the similar procedure to prove property \ref{Qual_De_Equal} in section \ref{ShadowQuadSettings}, we can obtain original problems and their decoupled problems are equivalent. 
\end{proof}
Based on the above properties, we know the shadow price for the decoupled problems \req{D_F} can be directly applied to sensitivity analysis for the DSO problem \req{DSO q}. Hence, we can first analyze the decoupled problems. Let $\underline{\boldsymbol{\lambda}}_i$, $\boldsymbol{\lambda}_i$ and $\overline{\boldsymbol{\lambda}}_i$ denote the optimal dual variable associated with constraint \req{pro m2} in problems \req{D_F1}, \req{D_F} and \req{D_F3} for prosumer $i$. Similar to the quadratic settings, we only focus on the targeted constraint i.e., the $j^{th}$ general linear constraint for prosumer $i$, \req{SuperConstraint}, and assume this constraint is tight at optimality while prosumer $i$'s other constraints are not (Similarly, there is no assumption about the other prosumers' constraints). Then derive the closed-form upper and lower bounds of the shadow price, as shown below.

\begin{theorem}[Dual Bounding Theorem]\label{thm:Dual Bounding} Assume for certain given $\underline{\textbf{r}}_k$ and $\overline{\textbf{r}}_k$, prosumer $k$'s optimization problems \req{D_F1}, \req{D_F} and \req{D_F3} achieve optimality with only the $j^{th}$ constraint of \req{pro m2}, i.e.,
	$\sum_{a \in\mathcal{A}_{k}^{\left(j\right)}} \sum_{t \in \mathcal{H}_{k,a}^{\left(j\right)}}\alpha_{k,a}^{\left(j\right)}\left(t\right) q_{k,a}^{D}\left(t\right)\leq [\textbf{h}_i]_j$ being tight. We denote the corresponding dual variables (i.e., shadow prices) as $\left[\underline{\boldsymbol{\lambda}}_k^*\right]_j$, $\left[\boldsymbol{\lambda}_k^*\right]_j$ and $[\overline{\boldsymbol{\lambda}}_k^*]_j$ respectively. Then, (1) $\left[\boldsymbol{\lambda}_k^*\right]_j$ is bounded by the following inequalities,
	\begin{align}
		&\max\{\left[\underline{\boldsymbol{\lambda}}_k^*\right]_j - H_{k,1},[\overline{\boldsymbol{\lambda}}_k^*]_j-H_{k,3}, 0\}\leq \left[\boldsymbol{\lambda}_k^*\right]_j \nonumber 
		\\
		&\leq \min\{\left[\underline{\boldsymbol{\lambda}}_k^*\right]_j+H_{k,1},[\overline{\boldsymbol{\lambda}}_k^*]_j+H_{k,3}\},\label{bounds}
	\end{align} 
	in which
	\begin{align}
		[\underline{\boldsymbol{\lambda}}_k^*]_j &= 
		\eta_k\Bigg[\sum_{a \in\mathcal{A}_{k}^{\left(j\right)}}\sum_{t \in \mathcal{H}_{k,a}^{\left(j\right)}}\alpha_{k,a}^{\left(j\right)}\left(t\right)\bigg[\mu_k\underline{r}_{k,a}(t)\nonumber
		\\
		&\quad\ \left. - \frac{\partial F_k}{\partial q_{k,a}(t)}\left(\underline{r}_{k,a}(t)\right)  - b_0(t)\right] - {\mu_k}[\textbf{h}_i]_j  \Bigg]^+,\label{lowerOptLamda}
		\\
		[\overline{\boldsymbol{\lambda}}_k^*]_j &=
		\eta_k\Bigg[\sum_{a \in\mathcal{A}_{k}^{\left(j\right)}}\sum_{t \in \mathcal{H}_{k,a}^{\left(j\right)}}\alpha_{k,a}^{\left(j\right)}\left(t\right)\bigg[L_k\overline{r}_{k,a}(t)\nonumber
		\\
		&\quad\ \left. - \frac{\partial F_k}{\partial q_{k,a}(t)}\left(\overline{r}_{k,a}(t)\right)  - b_0(t)\right] - {L_k}[\textbf{h}_i]_j  \Bigg]^+,\label{upperOptLamda}
		\\
		H_{k,1}&=\sqrt{\eta_k}\left[\mu_k\|\underline{\textbf{q}}_k^{D*}\|+L_k\|\textbf{q}_k^{D*}\|+\left(\mu_k+L_k\right)\|\underline{\textbf{r}}_k\|\right],\label{H1}
		\\
		H_{k,3}&=\sqrt{\eta_k}L_k\left(\|\overline{\textbf{q}}_k^{D*}\|+\|\textbf{q}_k^{D*}\|+2\|\overline{\textbf{r}}_k\|\right)\label{H3},
		\\
		\eta_k &= \bigg\{\sum_{a \in\mathcal{A}_{k}^{\left(j\right)}}\sum_{t \in \mathcal{H}_{k,a}^{\left(j\right)}} \left[\alpha_{k,a}^{\left(j\right)}(t)\right]^2\bigg\}^{-1};
	\end{align}
	\newline{(2)} if $\underline{\textbf{r}}_k^{*}$ and $\overline{\textbf{r}}_k^{*}$ are vectors such that $\underline{\textbf{q}}_k^{D*}=\textbf{q}_k^{D*}=\overline{\textbf{q}}_k^{D*}$, then the form of \eqref{bounds} can be represented with $H_{k,1}$ and $H_{k,3}$ changed to
	\begin{align}H_{k,1}&=\sqrt{\eta_k}(L_k+\mu_k)(\|\textbf{q}_k^{D*}\|+\|\underline{\textbf{r}}_k^{*}\|),\label{H1new}
		\\
		H_{k,3}&=2\sqrt{\eta_k}L_k(\|\textbf{q}_k^{D*}\|+\|\overline{\textbf{r}}_k^{*}\|).\label{H3new}
	\end{align}
\end{theorem}

\begin{proof}[Proof] (1) By first order optimality conditions of \req{D_F1}, \req{D_F} and \req{D_F3}, we have
	\begin{align}
		\mu_k \underline{\textbf{q}}_k^{D*}- \mu_k \underline{\textbf{r}}_k+\nabla G_k( \underline{\textbf{r}}_k)+A_k^{T}\underline{\boldsymbol{\lambda}}_k^{*}=0,\label{foc1}\\
		\nabla G_k( \textbf{q}_k^{D*})+A_k^{T}\boldsymbol{\lambda}_k^{*}=0,\label{foc2}\\
		L_k \overline{\textbf{q}}_k^{D*}-L_k \overline{\textbf{r}}_k+\nabla G_k( \overline{\textbf{r}}_k)+A_k^{T}\overline{\boldsymbol{\lambda}}_k^{*}=0.\label{foc3}
	\end{align}
	
	Subtracting \eqref{foc2} from \eqref{foc1} and \eqref{foc3} respectively, we obtain
	\begin{align}
		A_{k}^{T}(\underline{\boldsymbol{\lambda}}^{*}_k - \boldsymbol{\lambda}^{*}_k)=&\nabla F_k( \textbf{q}_{k}^{D*})-\nabla F_k( \underline{\textbf{r}}_k)-\mu_k \underline{\textbf{q}}_k^{D*}+\mu_k \underline{\textbf{r}}_k,\label{KKT extract1}\\
		A_k^{T}(\overline{\boldsymbol{\lambda}}^{*}_k - \boldsymbol{\lambda}^{*}_k)=&\nabla F_k( \textbf{q}_k^{D*})-\nabla F_k( \overline{\textbf{r}}_k)-L_k \overline{\textbf{q}}_k^{D*}+L_k \overline{\textbf{r}}_k. \label{KKT extract2}
	\end{align}
	
	By using triangle inequality and definition \ref{Def2} we have\\
	\begin{align*}
		&\|A_k^{T}(\underline{\boldsymbol{\lambda}}^{*}_k - \boldsymbol{\lambda}^{*}_k)\| \leq \|\nabla F_k( \textbf{q}_k^{D*})-\nabla F_k( \underline{\textbf{r}}_k)\|+\|\mu_k\underline{\textbf{q}}_k^{D*}\| \nonumber
		\\
		&+\|\mu_k\underline{\textbf{r}}_k\|
		\leq L_k\|\textbf{q}_k^{D*}-\underline{\textbf{r}}_k\|+\mu_k\|\underline{\textbf{q}}_k^{D*}\|+\mu_k\|\underline{\textbf{r}}_k\| \nonumber
		\\
		&\leq \mu_k\|\underline{\textbf{q}}_k^{D*}\|+L_k\|\textbf{q}_k^{D*}\|+(\mu_k+L_k)\|\underline{\textbf{r}}_k\|,
	\end{align*}
	\begin{align*}
		&\|A_k^{T}(\overline{\boldsymbol{\lambda}}^{*}_k - \boldsymbol{\lambda}^{*}_k)\| 
		\leq \|\nabla F_k( \textbf{q}_k^{D*})-\nabla F_k( \overline{\textbf{r}}_k)\|+\|L_k\overline{\textbf{q}}_k^{D*}\| \nonumber
		\\
		&+\|L_k\overline{\textbf{r}}_k\| 
		\leq L_k\|\textbf{q}_k^{D*}-\overline{\textbf{r}}_k\|+L_k\|\overline{\textbf{q}}_k^{D*}\|+L_k\|\overline{\textbf{r}}_k\|\nonumber
		\\
		&\leq L_k\left(\|\textbf{q}_k^{D*}\|+\|\overline{\textbf{q}}_k^{D*}\|+2\|\overline{\textbf{r}}_k\|\right).
	\end{align*}
	Since problems \req{D_F1}, \req{D_F} and \req{D_F3} achieve optimality with the same $j^{th}$ constraint tight, by complimentary slackness conditions we have
	\begin{align}
		\left[\underline{\boldsymbol{\lambda}}_k^*\right]_l=\left[\boldsymbol{\lambda}_k^*\right]_l=[\overline{\boldsymbol{\lambda}}_k^*]_l=0, \forall l \neq j,\label{compl}
	\end{align}
	The above three relations imply
	\begin{align}
		&\|A_k^{T}(\underline{\boldsymbol{\lambda}}^{*}_k - \boldsymbol{\lambda}^{*}_k)\|=\left|\left[\underline{\boldsymbol{\lambda}}_k^*\right]_j-\left[\boldsymbol{\lambda}_k^*\right]_j\right|\sqrt{\sum_{a \in\mathcal{A}_{k}^{\left(j\right)}}\sum_{t \in \mathcal{H}_{k,a}^{\left(j\right)}} \left[\alpha_{k,a}^{\left(j\right)}(t)\right]^2}\nonumber
		\\
		&\leq\mu_k\|\underline{\textbf{q}}_k^{D*}\|+L_k\|\textbf{q}_k^{D*}\|+(\mu_k+L_k)\|\underline{\textbf{r}}_k\|,
		\\
		&\|A_k^{T}(\overline{\boldsymbol{\lambda}}^{*}_k - \boldsymbol{\lambda}^{*}_k)\|=\left|[\overline{\boldsymbol{\lambda}}_k^*]_j-\left[\boldsymbol{\lambda}_k^*\right]_j\right|\sqrt{\sum_{a \in\mathcal{A}_{k}^{\left(j\right)}}\sum_{t \in \mathcal{H}_{k,a}^{\left(j\right)}} \left[\alpha_{k,a}^{\left(j\right)}(t)\right]^2}\nonumber
		\\
		&\leq L_k\left(\|\textbf{q}_k^{D*}\|+\|\overline{\textbf{q}}_k^{D*}\|+2\|\overline{\textbf{r}}_k\|\right).
	\end{align}
	Multiplying $\eta_k$ on both sides of the above two inequalities, we have
	\begin{align}
		\left|\left[\underline{\boldsymbol{\lambda}}_k^*\right]_j-\left[\boldsymbol{\lambda}_k^*\right]_j\right|
		\leq& \sqrt{\eta_k}\left[\mu_k\|\underline{\textbf{q}}_k^{D*}\|+L_k\|\textbf{q}_k^{D*}\|\right. \nonumber
		\\
		&\left.+\left(\mu_k+L_k\right)\|\underline{\textbf{r}}_k\|\right],\nonumber\\
		\left|[\overline{\boldsymbol{\lambda}}_k^*]_j-\left[\boldsymbol{\lambda}_k^*\right]_j\right|
		\leq& \sqrt{\eta_k} L_k\left(\|\textbf{q}_k^{D*}\|+\|\overline{\textbf{q}}_k^{D*}\|+2\|\overline{\textbf{r}}_k\|\right).\nonumber
	\end{align}
	Because $\left[\boldsymbol{\lambda}_k^*\right]_j \geq 0$ and definitions of $H_{k,1}$ and $H_{k,3}$, we obtain \eqref{bounds}.\\
	Next, we prove equations \req{lowerOptLamda} and \req{upperOptLamda} hold. Due to similarity, here we only provide the proof of \req{lowerOptLamda}. We know $\left[\underline{\boldsymbol{\lambda}}_k^*\right]_j$ is the shadow price associated with the $j^{th}$ general linear constraint in optimization problem \req{D_F1}. Note that the objective function of problem \req{D_F1}, $\underbar{$G$}_k(\underbar{$\textbf{r}$}_k,\textbf{q}_k^{D})$ has quadratic form in terms of $\textbf{q}_k^{D}$, as shown below.
	\begin{align}
		\underbar{$G$}_k(\underbar{$\textbf{r}$}_k, \textbf{q}_k^{D})=&\frac{\mu_k}{2}\left(\textbf{q}_k^{D}\right)^T \textbf{q}_k^{D} +\left[\nabla F(\underbar{$\textbf{r}$}_k) + \widetilde{\textbf{b}}_0- M \underbar{$\textbf{r}$}_k\right]^{T}\textbf{q}_k^{D} \nonumber
		\\
		&+ \left[F_k(\underbar{$\textbf{r}$}_k)-\nabla F_k(\underbar{$\textbf{r}$}_k)^{T}\underbar{$\textbf{r}$}_k+\frac{\mu_k}{2}\|\underbar{$\textbf{r}$}_k\|^{2}\right],\label{lowerBoundObjQuad}
	\end{align}
	where $\widetilde{\textbf{b}}_0 =   \left( \textbf{b}_{0}, \forall a \in \mathcal{A}_k\right)$.
	Following the similar procedure to derive \req{SuperShadow} in section \ref{DSOsec}, we have the shadow price $\left[\underline{\boldsymbol{\lambda}}_k^*\right]_j$ for problem \req{F1} as below.
	\begin{align}
		&\left[\underline{\boldsymbol{\lambda}}_k^*\right]_j = \left[\frac{\sum_{a \in\mathcal{A}_{k}^{\left(j\right)}}\sum_{t \in \mathcal{H}_{k,a}^{\left(j\right)}}\frac{\alpha_{k,a}^{\left(j\right)}(t)\underline{\pi}_{k,a}(t)}{-\mu_k} + [\textbf{h}_i]_j}{\sum_{a \in\mathcal{A}_{k}^{\left(j\right)}}\sum_{t \in \mathcal{H}_{k,a}^{\left(j\right)}} \frac{\left[\alpha_{k,a}^{\left(j\right)}(t)\right]^2}{-\mu_k}}\right]^+,\label{ComparisonShadow}
	\end{align}
	where $\underline{\pi}_{k,a}(t)= \mu_k\underline{r}_{k,a}(t)- \frac{\partial F_k}{\partial q_{k,a}(t)}\left(\underline{r}_{k,a}(t)\right) - b_0(t)$. Multiplying $-\mu_k$ on both of the numerator and the denominator in the right hand side of equality \req{ComparisonShadow}, we obtain
	\begin{align}
		&\left[\underline{\boldsymbol{\lambda}}_k^*\right]_j =\eta_k\Bigg[\sum_{a \in\mathcal{A}_{k}^{\left(j\right)}}\sum_{t \in \mathcal{H}_{k,a}^{\left(j\right)}}\alpha_{k,a}^{\left(j\right)}(t)\underline{\pi}_{k,a}(t) - {\mu_k}[\textbf{h}_i]_j \Bigg]^+,\nonumber
	\end{align}
	which is equivalent to \req{lowerOptLamda}.
	\newline{(2)} By property \ref{EqualSolutions}, we can substitute $\textbf{q}_k^{D*}=\underline{\textbf{q}}_k^{D*}, \underline{\textbf{r}}_k=\underline{\textbf{r}}_k^{*}$ into \eqref{H1} and $\textbf{q}_k^{D*}=\overline{\textbf{q}}_k^{D*},  \overline{\textbf{r}}_k=\overline{\textbf{r}}_k^{*}$ into \eqref{H3}. Then we obtain \req{H1new} and \req{H3new}. This completes the proof.
\end{proof}
Based on the above theorem, we can substitute different values of $\underbar{$\textbf{r}$}_k$ and $\overline{\textbf{r}}_k$ to obtain various bounds of the shadow price $\left[\boldsymbol{\lambda}_k^*\right]_j$. The pair $\underline{\textbf{r}}_k^{*}= \textbf{q}_k^{D*}$ and $\overline{\textbf{r}}_k^{*} = \textbf{q}_k^{D*}$ would give us the exact bounds, but cannot provide us too much information due to terms cancelation. Another candidate is $\underline{\textbf{r}}_k=\bar{\textbf{r}}_k=0$, by which we can obtain the following simpler upper and lower bounds of the shadow price $\left[\boldsymbol{\lambda}_k^*\right]_j$. For the rest of the paper, we use $\underline{\textbf{q}}_k^{D*}$ and $\overline{\textbf{q}}_k^{D*}$ to denote the optimal solutions with $\underline{\textbf{r}}_k=\bar{\textbf{r}}_k=0$.

\begin{corollary}\label{cor:Dual Bounding1} In Theorem \ref{thm:Dual Bounding} (1), if $\underline{\textbf{r}}_k=\bar{\textbf{r}}_k=0$, then the lower bound of $\left[\boldsymbol{\lambda}_k^*\right]_j$, denoted by $\left[\boldsymbol{\lambda}_k^*\right]_j^{lower}$, could be further simplified as
	\begin{align}
	\max\Bigg\{&\eta_k\Bigg[\sum_{a \in\mathcal{A}_{k}^{\left(j\right)}}\sum_{t \in \mathcal{H}_{k,a}^{\left(j\right)}}\alpha_{k,a}^{\left(j\right)}(t)\left[-\frac{\partial F_k}{\partial q_{k,a}(t)}\left(0\right) - b_0(t)\right] \nonumber
	\\
	& - {\mu_k}[\textbf{h}_k]_j \Bigg]^+ -\sqrt{\eta_k}\left(\mu_k\|\underline{\textbf{q}}_k^{D*}\|+L_k\|\textbf{q}_k^{D*}\|\right),\nonumber
	\\
	&\eta_k\Bigg[\sum_{a \in\mathcal{A}_{k}^{\left(j\right)}}\sum_{t \in \mathcal{H}_{k,a}^{\left(j\right)}}\alpha_{k,a}^{\left(j\right)}(t)\left[-\frac{\partial F_k}{\partial q_{k,a}(t)}\left(0\right) - b_0(t)\right] \nonumber
	\\
	& - {L_k}[\textbf{h}_k]_j \Bigg]^+ -\sqrt{\eta_k}L_k\left(\|\overline{\textbf{q}}_k^{D*}\|+\|\textbf{q}_k^{D*}\|\right), \  0\Bigg\}; \label{r0NewLower}
	\end{align}
	the upper bound $\left[\boldsymbol{\lambda}_k^*\right]_j^{upper}$ of $\left[\boldsymbol{\lambda}_k^*\right]_j$ could be simplified as
	\begin{align}
	\min\Bigg\{&\eta_k\Bigg[\sum_{a \in\mathcal{A}_{k}^{\left(j\right)}}\sum_{t \in \mathcal{H}_{k,a}^{\left(j\right)}}\alpha_{k,a}^{\left(j\right)}(t)\left[-\frac{\partial F_k}{\partial q_{k,a}(t)}\left(0\right) - b_0(t)\right] \nonumber
	\\
	& - {\mu_k}[\textbf{h}_k]_j \Bigg]^+ +\sqrt{\eta_k}\left(\mu_k\|\underline{\textbf{q}}_k^{D*}\|+L_k\|\textbf{q}_k^{D*}\|\right),\nonumber
	\\
	&\eta_k\Bigg[\sum_{a \in\mathcal{A}_{k}^{\left(j\right)}}\sum_{t \in \mathcal{H}_{k,a}^{\left(j\right)}}\alpha_{k,a}^{\left(j\right)}(t)\left[-\frac{\partial F_k}{\partial q_{k,a}(t)}\left(0\right) - b_0(t)\right] \nonumber
	\\
	& - {L_k}[\textbf{h}_k]_j \Bigg]^+ +\sqrt{\eta_k}L_k\left(\|\overline{\textbf{q}}_k^{D*}\|+\|\textbf{q}_k^{D*}\|\right)\Bigg\}.\label{r0NewUpper}
	\end{align}
\end{corollary}
We notice that the bounds are related to prosumer $k$'s initial utility increasing rate $- \frac{\partial F_k}{\partial q_{k,a}(t)}\left(0\right)$, which is positive by Assumption \ref{convexAssumption}.

\begin{corollary}\label{cor:Dual Bounding2} In theorem \ref{thm:Dual Bounding} (1), if $\underline{\textbf{r}}_k=\bar{\textbf{r}}_k=0$ and $[\textbf{h}_k]_j = 0$, then we can further simplify the bounds in Corollary \ref{cor:Dual Bounding1}, where the lower bound $\left[\boldsymbol{\lambda}_k^*\right]_j^{lower}$ of $\left[\boldsymbol{\lambda}_k^*\right]_j$ becomes
	\begin{align}
	&\left[\eta_k\Bigg[\sum_{a \in\mathcal{A}_{k}^{\left(j\right)}}\sum_{t \in \mathcal{H}_{k,a}^{\left(j\right)}}\alpha_{k,a}^{\left(j\right)}(t)\left[- \frac{\partial F_k}{\partial q_{k,a}(t)}\left(0\right) - b_0(t)  \right]  \Bigg]^+\right. \nonumber 
	\\
	&\ \left.   -     \sqrt{\eta_k}L_k\left(\|\overline{\textbf{q}}_k^{D*}\|+\|\textbf{q}_k^{D*}\|\right)\right]^+;\label{h0NewLower}
	\end{align}
	the upper bound $\left[\boldsymbol{\lambda}_k^*\right]_j^{upper}$ of $\left[\boldsymbol{\lambda}_k^*\right]_j$ becomes
	\begin{align}
	&\eta_k\Bigg[\sum_{a \in\mathcal{A}_{k}^{\left(j\right)}}\sum_{t \in \mathcal{H}_{k,a}^{\left(j\right)}}\alpha_{k,a}^{\left(j\right)}(t)\left[- \frac{\partial F_k}{\partial q_{k,a}(t)}\left(0\right) - b_0(t) \right]  \Bigg]^+ \nonumber
	\\
	&\ +\sqrt{\eta_k}L_k\left(\|\overline{\textbf{q}}_k^{D*}\|+\|\textbf{q}_k^{D*}\|\right).\label{h0NewUpper}
	\end{align}
\end{corollary}
\begin{proof} By observing \eqref{KKT extract1} and \eqref{KKT extract2} and specifying $\underline{\textbf{r}}=\overline{\textbf{r}}=0$, we can obtain
	\begin{align*}
	\mu_k \underline{\textbf{q}}_k^{D*} &= \nabla F_k( \textbf{q}_{k}^{D*})-\nabla F_k(0)-A_{k}^{T}(\underline{\boldsymbol{\lambda}}^{*}_k-\boldsymbol{\lambda}^{*}_k),\\
	L_k \overline{\textbf{q}}_k^{D*} &= \nabla F_k( \textbf{q}_k^{D*})-\nabla F_k(0) - A_k^{T}(\overline{\boldsymbol{\lambda}}^{*}_k - \boldsymbol{\lambda}^{*}_k).
	\end{align*}
	Substituting $\underline{\textbf{r}}_k=\overline{\textbf{r}}_k=0$ and $[\textbf{h}_k]_j = 0$ into \req{lowerOptLamda} and \req{upperOptLamda}, we can have $\left[\underline{\boldsymbol{\lambda}}_k^*\right]_j = [\overline{\boldsymbol{\lambda}}_k^*]_j$. Combining with \req{compl}, we obtain $\underline{\boldsymbol{\lambda}}^*_k = \overline{\boldsymbol{\lambda}}^*_k$, which leads to $\mu_k \underline{\textbf{q}}_k^{D*} = L_k \overline{\textbf{q}}_k^{D*}$ and thus $\mu_k\|\underline{\textbf{q}}_k^{D*}\| = L_k\|\overline{\textbf{q}}_k^{D*}\|$. By directly substituting this equation and $[\textbf{h}_k]_j = 0$ into \eqref{r0NewLower} and \eqref{r0NewUpper}, we can obtain \req{h0NewLower} and \req{h0NewUpper}.
\end{proof}
In the general convex settings, the closed-form upper and lower bounds of the shadow price in the above Dual Bounding Theorem \ref{thm:Dual Bounding}, its Corollaries \ref{cor:Dual Bounding1} and \ref{cor:Dual Bounding2}, and Property \ref{Conv_De_Property} provide foundations for GSAA, which can be used as a measuring tool to estimate the ranges of prosumer $k$'s contribution potential on social welfare with change of the resource capacity in DR. More specifically, interval $\left[K\left[\boldsymbol{\lambda}_k^*\right]_j^{lower}, K\left[\boldsymbol{\lambda}_k^*\right]_j^{upper}\right]$ shows the potential range of prosumer $k$'s contribution on social welfare by increasing his/her resource capacity $[\textbf{h}_k]_j$ by $K$ units. Similar to the quadratic setting, we next answer the question of which prosumer should be allocated the additional resource capacity to maximize the improvement in social welfare, when there is a limited budget in expanding resource capacity. We do so by comparing these potential ranges. When one's potential range is strictly higher than the other's, then that prosume should get priority. In case there is overlap, then the result may be indecisive. In such case, further analysis is required.

\subsection{GSAA Applications in General Convex Settings}
Similar to section \ref{GSAA-Q app}, here we provide three application examples: enlarging AC comfort zone, allowing net-selling behavior and allowing EVs to discharge. However, different from section \ref{GSAA-Q app}, in general convex settings GSAA only gives us closed-form upper and lower bounds of the shadow price instead of the closed-form shadow price itself. Nevertheless, under certain conditions, we can identify the prosumer with most significant impact on social welfare based on these bounds. 
\subsubsection{GSAA for Enlarging AC Comfort Zone} 
Prosumer $k$'s AC (labeled by $AC$) comfort zone constraint is given by \req{ACoriginal}. Take the right inequality of \req{ACoriginal}, i.e., the upper limit constraint, as an example. Similar to Section \ref{GSAA_AC_Q}, we can reformulate it into the $j^{th}$ general linear constraint for prosumer $k$. By using GSAA, i.e., substituting specific parameters of the constraint into Corollary \ref{cor:Dual Bounding1}, we can obtain upper and lower bounds of the shadow price for increasing AC comfort zone's upper limit, where the lower bound $\underline{\lambda}_{AC}^*$ of the shadow price $\lambda_{AC}^*$ is
\begin{align*}
&\max\Bigg\{\eta_{k}\Bigg[\sum_{\tau = 1}^{t}\alpha_{k,AC}(\tau)\left(- \frac{\partial F_k}{\partial q_{k,AC}(\tau)}\left(0\right) - b_0(\tau)  \right) \nonumber\\
&- {\mu_k}\overline{Q}_{k,AC}^D   \Bigg]^+ -\sqrt{\eta_{k}} \left(\mu_k\|\underline{\textbf{q}}_k^{D*}\| + L_k\|\textbf{q}_k^{D*}\|\right),\nonumber \\ 
&\eta_{k}\Bigg[\sum_{\tau = 1}^{t}\alpha_{k,AC}(\tau)\left(- \frac{\partial F_k}{\partial q_{k,AC}(\tau)}\left(0\right) - b_0(\tau)  \right) \nonumber\\
&- {L_k}\overline{Q}_{k,AC}^D   \Bigg]^+ -\sqrt{\eta_{k}} L_k\left(\|\overline{\textbf{q}}_k^{D*}\|+\|\textbf{q}_k^{D*}\|\right), 0\Bigg\},
\end{align*}
where $\eta_{k} = \left[\sum_{\tau = 1}^{t} \alpha_{k,AC}^2(\tau)\right]^{-1}$; the upper bound $\overline{\lambda}_{AC}^*$ of the shadow price is
\begin{align*}
&\min\Bigg\{\eta_{k}\Bigg[\sum_{\tau = 1}^{t}\alpha_{k,AC}(\tau)\left(- \frac{\partial F_k}{\partial q_{k,AC}(\tau)}\left(0\right) - b_0(\tau)  \right) \nonumber\\
&- {\mu_k}\overline{Q}_{k,AC}^D   \Bigg]^+ +\sqrt{\eta_{k}} \left(\mu_k\|\underline{\textbf{q}}_k^{D*}\| + L_k\|\textbf{q}_k^{D*}\|\right),\nonumber \\ 
&\eta_{k}\Bigg[\sum_{\tau = 1}^{t}\alpha_{k,AC}(\tau)\left(- \frac{\partial F_k}{\partial q_{k,AC}(\tau)}\left(0\right) - b_0(\tau)  \right) \nonumber\\
&- {L_k}\overline{Q}_{k,AC}^D   \Bigg]^+ +\sqrt{\eta_{k}} L_k\left(\|\overline{\textbf{q}}_k^{D*}\|+\|\textbf{q}_k^{D*}\|\right)\Bigg\}.
\end{align*}

We note that unlike \req{sen1} in the quadratic settings, the above bounds do not need information about the specific forms or coefficients of the net utility function. Instead, it requires several key values to complete the calculation, the initial utility increasing rate $- \frac{\partial F_k}{\partial q_{k,AC}(\tau)}\left(0\right)$, the optimal solutions of problems (I) and (III), i.e., $\underline{\textbf{q}}_k^{D*}$ and $\overline{\textbf{q}}_k^{D*}$, and the operating point $\textbf{q}_k^{D*}$\footnote{We could obtain the initial utility increasing rate by learning from historical data or by surveys or tests for prosumers, obtain optimal solutions of problems (I) and (III) by solving quadratic programming problems, and obtain the operating point by direct measuring or reading from meters. The implementation details, however, are out the scope of this paper. }. 

For two prosumers with similar utility functions and similar AC parameters $(\alpha_{i,AC}\left(\tau\right))$, everything else equal, the one accepts a higher temperature, reflected as larger $\overline{Q}_{i,AC}^D$, has a relatively lower range of contribution potential, just as in Section \ref{GSAA_AC_Q}. However, in this case, it is harder to guarantee all other parameters are equal, as the ranges depends on $\overline{\textbf{q}}_k^{D*}$ and $\underline{\textbf{q}}_k^{D*}$. 

\subsubsection{GSAA for Allowing Net Selling }
Allowing prosumer $k$ (a net buyer) to be a net-seller in DR is equivalent to relaxing the net buying constraint \req{NetSeller}, which can be reformulated as the $j$th general linear constraint for prosumer $k$ (similar to Section \ref{GSAA_NS_Q}). By using GSAA i.e., substituting specific parameters of the constraint to Corollary \ref{cor:Dual Bounding2}, we can obtain the lower bound $\underline{\lambda}_{NS}^*$ of the shadow price $\lambda_{NS}^*$ for allowing net selling as
\begin{align*}
\left[\frac{\left[\sum_{a \in\mathcal{A}_{k}}\left(b_0(t)+\frac{\partial F_k}{\partial q_{k,a}(t)}\left(0\right) \right)\right]^+}{|\mathcal{A}_k|}-\frac{L_k\left(\|\overline{\textbf{q}}_k^{D*}\|+\|\textbf{q}_k^{D*}\|\right)}{\sqrt{|\mathcal{A}_{k}|}}\right]^+;
\end{align*}
and the upper bound $\overline{\lambda}_{NS}^*$ of the shadow price as
\begin{align*}
&\frac{\left[\sum_{a \in\mathcal{A}_{k}}\left(b_0(t)+\frac{\partial F_k}{\partial q_{k,a}(t)}\left(0\right) \right)\right]^+}{|\mathcal{A}_k|}+\frac{L_k\left(\|\overline{\textbf{q}}_k^{D*}\|+\|\textbf{q}_k^{D*}\|\right)}{\sqrt{|\mathcal{A}_{k}|}}.
\end{align*}

From the above bounds, we can observe that the effect of allowing net selling on social welfare improvement is not due to a single appliance. Instead, it is affected by a combination of all appliances, especially their initial utility increasing rates $-\frac{\partial F_k}{\partial q_{k,a}(t)}\left(0\right)$, which implies the high complexity of accurate sensitivity analysis.

For two prosumers, if both of their average initial utility increasing rates are no more than the utility company's cost and all other properties are similar (i.e., similar amounts of appliances, $\textbf{q}_k^{D*}$ , $L_k$ and $\overline{\textbf{q}}_k^{D*}$), then the one with a lower average initial utility increasing rate has a relatively higher range of contribution potential on social welfare by allowing them to net sell K units of energy.
\subsubsection{Effects of Allowing EVs to Discharge} 
Allowing prosumer $k$'s EV (labeled by $EV$) to discharge at time period $t$ is equivalent to relaxing the constraint $q_{i,a}^D\left(t\right) \geq 0$. Similarly, the constraint can be reformulated as the $j$th general linear constraint for prosumer $k$ (see Section \ref{GSAA_EV_Q}). To obtain the upper and lower bounds of the shadow price $\lambda_{EV}^*$ associated with the constraint, we can use GSAA, i.e., substituting parameters of the constraint into Corollary \ref{cor:Dual Bounding2}. The lower bound  $\underline{\lambda}_{EV}^*$ becomes
\begin{align}
\left[\left[b_0(t)+\frac{\partial F_k}{\partial q_{k,EV}(t)}\left(0\right)\right]^+ -L_k\left(\|\overline{\textbf{q}}_k^{D*}\|+\|\textbf{q}_k^{D*}\|\right)\right]^+;\label{EVlower}
\end{align}
the upper bound $\overline{\lambda}_{EV}^*$ has the following form
\begin{align}
&\left[b_0(t)+\frac{\partial F_k}{\partial q_{k,EV}(t)}\left(0\right)\right]^+ +L_k\left(\|\overline{\textbf{q}}_k^{D*}\|+\|\textbf{q}_k^{D*}\|\right).\label{EVupper}
\end{align}

We can see that the above shadow price is monotonically increasing in the current utility company's production cost and monotonically decreasing against the prosumer $k$'s initial utility increasing rate (when other parameters are fixed).

For two prosumers with similar $\textbf{q}_k^{D*}$ , $L_k$ and $\overline{\textbf{q}}_k^{D*}$, the one with a lower initial utility increasing rate $-\frac{\partial F_k}{\partial q_{k,EV}(t)}\left(0\right)$ (not exceeding the utility company's cost) has a relatively higher range of contribution potential.

\subsection{Case Study: Two Prosumers with EVs }\label{Case Study}
Consider a simple one-time-period case (i.e., $t \in \mathcal{H}=\{1\}$), where we want to compare two prosumers, $l$ and $k$, on their potential to benefit social welfare by discharging their EVs. We assume their negative net utility functions are related by a constant factor, i.e.,
\begin{align}
F_l\left(\cdot \right) = \beta F_k\left(\cdot \right), \ (\beta \in \mathbb{R}^+),\label{scaledUtility}
\end{align}
and at optimal solution of their respective (\ref{opt:Prosumer}) problems, only EV discharging constraints $-q_{k,EV}^{D}\left(t\right)\leq 0$ for both $k$ and $l$ are tight. We next introduce a corollary based on \eqref{scaledUtility}.
\begin{corollary}
	If prosumer $l$'s and prosumer $k$'s negative net utility functions satisfy
	$F_l\left(\cdot \right) = \beta F_k\left(\cdot \right)$, where  $\beta \in \mathbb{R}^+, F_k\left(\cdot \right)$ is $\mu_k$-strongly convex and has $L_k$-Lipschitz continuous gradient and $\underline{r}_l=\overline{r}_l=0$ satisfy the assumption in Theorem \ref{thm:Dual Bounding} (1), i.e., for prosumer $l$, only the $g^{th}$ general linear constraint being tight at the optimality of problems \req{D_F1}, \req{D_F} and \req{D_F3}, then the optimal solutions $\underline{q}_l^{D*}$ and $\overline{q}_l^{D*}$ for prosumer $l$'s problems \req{D_F1} and \req{D_F3}  have the following forms
	\begin{gather}
	\underline{\textbf{q}}_l^{D*} =  - \frac{\nabla F_k( 0)}{\mu_k} - \frac{\textbf{b}_{0,l}}{\beta \mu_k} - \frac{\textbf{e}_l \left[\underline{\boldsymbol{\lambda}}_l^*\right]_g}{\beta \mu_k},\label{Lem_qLower}\\
	\overline{\textbf{q}}_l^{D*} =  - \frac{\nabla F_k( 0)}{L_k} - \frac{\textbf{b}_{0,l}}{\beta L_k} - \frac{\textbf{e}_l [\overline{\boldsymbol{\lambda}}_l^*]_g}{\beta L_k}.\label{Lem_qUpper}
	\end{gather}
	where $\textbf{b}_{0,l} = \left(\textbf{b}_0, \forall a \in \mathcal{A}_l \right)$; $\left[\underline{\boldsymbol{\lambda}}_l^*\right]_g$ and $[\overline{\boldsymbol{\lambda}}_l^*]_g$ are shadow prices associated with the $g^{th}$ general linear constraint for prosumer $l$ $\sum_{a \in\mathcal{A}_{l}^{\left(g\right)}} \sum_{t \in \mathcal{H}_{l,a}^{\left(g\right)}}$ $\alpha_{l,a}^{\left(g\right)}\left(t\right)q_{l,a}^{D}\left(t\right) \leq [\textbf{h}_l]_g$,
	\begin{align}
	\left[\underline{\boldsymbol{\lambda}}_l^*\right]_g =&\theta_l\Bigg[\sum_{a \in\mathcal{A}_{l}^{\left(g\right)}}\sum_{t \in \mathcal{H}_{l,a}^{\left(g\right)}}\alpha_{l,a}^{\left(g\right)}(t)\left[ - \beta \frac{\partial F_k}{\partial q_{l,a}(t)}\left(0\right) - b_0(t)  \right] \nonumber
	\\
	&- {\beta\mu_k}[\textbf{h}_l]_g  \Bigg]^+, \label{lamda_l_lower}
	\\
	[\overline{\boldsymbol{\lambda}}_l^*]_g =& \theta_l\Bigg[\sum_{a \in\mathcal{A}_{l}^{\left(g\right)}}\sum_{t \in \mathcal{H}_{l,a}^{\left(g\right)}}\alpha_{l,a}^{\left(g\right)}(t)\left[ - \beta \frac{\partial F_k}{\partial q_{l,a}(t)}\left(0\right) - b_0(t)  \right] \nonumber
	\\
	&- \beta {L_k}[\textbf{h}_l]_g  \Bigg]^+ \label{lamda_l_upper},
	\\
	\theta_l =&\bigg\{\sum_{a \in\mathcal{A}_{l}^{\left(g\right)}}\sum_{t \in \mathcal{H}_{l,a}^{\left(g\right)}} \left[\alpha_{l,a}^{\left(g\right)}(t)\right]^2\bigg\}^{-1},\nonumber
	\end{align}
	and,
	\begin{gather}
	\textbf{e}_l =  \left( \left( \widetilde{\alpha}_{l,a}^{\left(g\right)}\left(t\right), \forall t \right), \forall a \in \mathcal{A}_l^{\left(g\right)} \right),\nonumber
	\end{gather}
	for which
	\begin{gather}
	\widetilde{\alpha}_{l,a}^{\left(g\right)}\left(t\right) = \begin{cases}
	\alpha_{l,a}^{\left(g\right)}\left(t\right)  & \text{if } t \in \mathcal{H}_{l,a}^{\left(g\right)} \\
	0  & \text{otherwise } 
	\end{cases}.\nonumber
	\end{gather}
\end{corollary}

\begin{proof}
	Since $F_k\left(\cdot \right)$ is $\mu_k$-strongly convex and has $L_k$-Lipschitz continuous gradient, $F_l\left(\cdot \right) = \beta F_k\left(\cdot \right)$ is $\beta \mu_k$-strongly convex and has $\beta L_k$-Lipschitz continuous gradient. By first order optimality condition of problem \req{D_F1} for prosumer $l$, we have
	\begin{align}
	\beta \mu_k \underline{\textbf{q}}^{D*}- \beta \mu_k \underline{\textbf{r}}+\nabla G_l( \underline{\textbf{r}})+A_l^{T}\underline{\boldsymbol{\lambda}}_l^{*}=0.\nonumber
	\end{align}
	Substituting $\underline{\textbf{r}}_l=0$ to the above equation, we have
	\begin{align*}
	\beta \mu_k \underline{\textbf{q}}_l^{D*} + \nabla F_l(0) + \textbf{b}_{0,l} +A_l^{T}\underline{\boldsymbol{\lambda}}_l^{*}=0.
	\end{align*}
	Using $F_l\left(\cdot \right) = \beta F_k\left(\cdot \right)$, we obtain
	\begin{align}
	\beta \mu_k \underline{\textbf{q}}_l^{D*} + \beta \nabla F_k(0) + \textbf{b}_{0,l} +A_l^{T}\underline{\boldsymbol{\lambda}}_l^{*}=0.\label{FOC_betaMu}
	\end{align}
	By complimentary slackness conditions, we have 
	\begin{align}
	A_l^{T}\underline{\boldsymbol{\lambda}}_l^{*} = \textbf{e}_l \left[\underline{\boldsymbol{\lambda}}_l^*\right]_g.\label{compl2}
	\end{align}
	Substituting \eqref{compl2} to \eqref{FOC_betaMu} and with some basic calculations, we can obtain \req{Lem_qLower}. Finally, using \eqref{lowerOptLamda} in Dual Bounding Theorem \ref{thm:Dual Bounding}, we can obtain \eqref{lamda_l_lower} by replacing subscript  $j$ with $g$, replacing $k$ with $l$, then substituting $\underline{\textbf{r}}_l=0$,  $\nabla F_l(0) = \beta \nabla F_k(0)$ and $\mu_l = \beta \mu_k$. The proof for \eqref{Lem_qUpper} is omitted since the process is similar.
\end{proof}
We will next focus on analyzing the constraint of the kind $-q_{k,EV}^{D}\left(t\right)\leq 0$, which corresponds to the $r^{th}$ constraint for prosumer $k$ and $g^{th}$ for prosumer $l$.

\subsubsection{Example ($\beta > 1$)}
Prosumer $k$ can contribute more on social welfare than prosumer $l$ in EV discharging, even when $k$'s utility is less than $l$'s utility. This case is shown by the following lemma.
\begin{lemma}\label{ExLemma1}
	If $\beta$ satisfies the following condition
	\begin{gather}
	1 < \frac{ - \nabla F_k(0) + L_k \| \textbf{q}_k^{D*} \| } { - \nabla F_k(0) - L_k \| \textbf{q}_l^{D*} \| }  < \beta < \frac{ b_0(t) } { - \nabla F_k(0)}, \label{example1}
	\end{gather}
	then the shadow price $\left[\boldsymbol{\lambda}_k^*\right]_r$ associated with $-q_{k,a}^{D}\left(t\right)\leq 0$ (the $g^{th}$ general linear constraint for prosumer $k$), is larger than the shadow price $\left[\boldsymbol{\lambda}_l^*\right]_g$ associated with $-q_{l,a}^{D}\left(t\right)\leq 0$ (the $r^{th}$ general linear constraint for prosumer $l$), i.e., $\left[\boldsymbol{\lambda}_k^*\right]_r > \left[\boldsymbol{\lambda}_l^*\right]_g$.
\end{lemma}

\begin{proof}
	We prove the inequality by showing the upper bound of $\left[\boldsymbol{\lambda}_l^*\right]_g$, $\left[\boldsymbol{\lambda}_l^*\right]_g^{upper}$ is less than the lower bound of $\left[\boldsymbol{\lambda}_k^*\right]_r$, $\left[\boldsymbol{\lambda}_k^*\right]_r^{lower}$. First, from \eqref{EVupper} and \eqref{scaledUtility} we can obtain
	\begin{align}
	&\left[\boldsymbol{\lambda}_l^*\right]_g^{upper} = \left[b_0(t)+ \beta \frac{\partial F_k}{\partial q_{k,a}(t)}\left(0\right)\right]^+ \nonumber
	\\
	&+ \beta L_k\left(\|\overline{\textbf{q}}_l^{D*}\|+\|\textbf{q}_l^{D*}\|\right).\label{l_EVlower}
	\end{align}
	Substituting \eqref{Lem_qUpper} to \eqref{l_EVlower}, we have
	\begin{align}
	&\left[\boldsymbol{\lambda}_l^*\right]_g^{upper} = \left[b_0(t)+ \beta \frac{\partial F_k}{\partial q_{k,a}(t)}\left(0\right)\right]^+ \nonumber
	\\
	&+\beta L_k \left\|- \frac{\nabla F_k( 0)}{L_k} - \frac{\textbf{b}_{0,l}}{\beta L_k} - \frac{\textbf{e}_l [\overline{\boldsymbol{\lambda}}_l^*]_g}{\beta L_k} \right\| + \beta L_k \|\textbf{q}_l^{D*}\| \nonumber.
	\end{align}
	Since we are considering the single-time-period case, $\frac{\partial F_k}{\partial q_{k,a}(t)}\left(0\right) =  \nabla F_k( 0) $ and $\textbf{b}_{0,l} = b_0(t)$. Also because the constraint for the shadow price $\left[\boldsymbol{\lambda}_l^*\right]_g$ is $-q_{l,a}^{D}\left(t\right)\leq 0$, we have $\textbf{e}_l =  -1$. Hence, the above equation becomes
	\begin{align}
	&\left[\boldsymbol{\lambda}_l^*\right]_g^{upper} = \left[b_0(t)+ \beta  \nabla F_k( 0) \right]^+ \nonumber
	\\
	&+ \left\|- \beta \nabla F_k( 0) - b_{0}(t)+ [\overline{\boldsymbol{\lambda}}_l^*]_g \right\| + \beta L_k \|\textbf{q}_l^{D*}\|.\label{l_EVlower_medium2}
	\end{align}
	By simplifying \eqref{lamda_l_upper}, we know $[\overline{\boldsymbol{\lambda}}_l^*]_g = \left[b_0(t)+ \beta  \nabla F_k( 0) \right]^+$. Therefore, \eqref{l_EVlower_medium2} becomes
	\begin{align}
	\left[\boldsymbol{\lambda}_l^*\right]_g^{upper} =& \left[b_0(t)+ \beta  \nabla F_k( 0) \right]^+ +\left\|- \beta \nabla F_k( 0) - b_{0}(t)\right.\nonumber
	\\
	&\left.+ \left[b_0(t)+ \beta  \nabla F_k( 0) \right]^+ \right\| + \beta L_k \|\textbf{q}_l^{D*}\| .\label{l_EVlower_medium3}
	\end{align}
	Notice that $ - \nabla F_k(0)$ represents the initial utility increasing rate, which is positive by Assumption \ref{convexAssumption}. Hence, from \eqref{example1} we know
	\begin{gather}
	b_0(t) + \beta  \nabla F_k(0) > 0.\label{positiveUpperShadow_l}
	\end{gather}
	Having the above inequality, \eqref{l_EVlower_medium3} can be simplified to
	\begin{gather}
	\left[\boldsymbol{\lambda}_l^*\right]_g^{upper} = b_0(t)+ \beta  \nabla F_k( 0) + \beta L_k \|\textbf{q}_l^{D*}\| .\label{l_EVlower_final}
	\end{gather}
	Next, we can obtain prosumer $k$'s shadow price lower bound by using \eqref{EVlower} and \eqref{Lem_qUpper} (with letting $\beta = 1$ and replacing subscript $l$ by $k$ and $g$ by $r$). 
	\begin{align}
	\left[\boldsymbol{\lambda}_k^*\right]_r^{lower} = &\bigg[\left[b_0(t)+\nabla F_k( 0)\right]^+ -L_k \left\|- \frac{\nabla F_k( 0)}{L_k}  \right.\nonumber 
	\\
	&\left. - \frac{\textbf{b}_{0,k}}{ L_k}- \frac{\textbf{e}_k [\overline{\boldsymbol{\lambda}}_k^*]_r}{ L_k}\right\|- L_k \|\textbf{q}_k^{D*}\|\bigg]^+ .\nonumber
	\end{align}
	Similar to the process we analyzed the coefficients $\textbf{b}_{0,l}$ and $\textbf{e}_l$ for $\left[\boldsymbol{\lambda}_l^*\right]_g^{upper}$, we can obtain $\textbf{b}_{0,k} = b_0(t)$, $\textbf{e}_k =  -1$ and $[\overline{\boldsymbol{\lambda}}_k^*]_r = \left[b_0(t)+  \nabla F_k( 0) \right]^+$. Hence, we have
	\begin{align}
	\left[\boldsymbol{\lambda}_k^*\right]_r^{lower} = &\bigg[\left[b_0(t)+\nabla F_k( 0)\right]^+ -\left\|- \nabla F_k( 0) - b_{0}(t)\right.\nonumber
	\\
	&\left.+ \left[b_0(t)+  \nabla F_k( 0) \right]^+\right\|- L_k \|\textbf{q}_k^{D*}\| \bigg]^+.\label{l_EVupper__medium}
	\end{align}
	Recall that $\beta > 1$. Therefore, from \eqref{positiveUpperShadow_l}, we have
	\begin{gather}
	b_0(t) +  \nabla F_k(0) > 0.\label{positiveUpperShadow_k}
	\end{gather}
	Having the above inequality, \eqref{l_EVupper__medium} can be simplified to
	\begin{gather}
	\left[\boldsymbol{\lambda}_k^*\right]_r^{lower} = \left[b_0(t)+  \nabla F_k( 0) -  L_k \|\textbf{q}_l^{D*}\|\right]^+.
	\end{gather}
	Last, by \eqref{example1} and $- \nabla F_k(0) + L_k \| \textbf{q}_l^{D*} \| > 0$, we have $- \nabla F_k(0) - L_k \| \textbf{q}_l^{D*} \| > 0$. Hence, from \eqref{example1}, we can obtain
	\begin{gather}
	\nabla F_k(0) - L_k \| \textbf{q}_k^{D*} \| >  \beta ( \nabla F_k(0) + L_k \| \textbf{q}_l^{D*} \|). \nonumber
	\end{gather}
	Adding $b_0(t)$ on both sides of the above inequality, we have
	\begin{align}
	&b_0(t) + \nabla F_k(0) - L_k \| \textbf{q}_k^{D*} \| \nonumber
	\\
	&> b_0(t) + \beta \nabla F_k(0) + \beta L_k \| \textbf{q}_l^{D*} \|. \label{example1_medium2}
	\end{align}
	Recall $b_0(t)$ is positive. Also by \eqref{positiveUpperShadow_l}, we have the right hand side of the above inequality is positive. Therefore, the left hand side is also positive, which implies
	\begin{align}
	&\left[\boldsymbol{\lambda}_k^*\right]_r^{lower} = b_0(t)+  \nabla F_k( 0) -  L_k \|\textbf{q}_k^{D*}\| \nonumber
	\\
	&> b_0(t) + \beta \nabla F_k(0) + \beta L_k \| \textbf{q}_l^{D*} \| = \left[\boldsymbol{\lambda}_l^*\right]_g^{upper} .
	\end{align}
	Hence, we obtain $\left[\boldsymbol{\lambda}_k^*\right]_r \geq \left[\boldsymbol{\lambda}_k^*\right]_r^{lower}  > \left[\boldsymbol{\lambda}_l^*\right]_g^{upper} \geq \left[\boldsymbol{\lambda}_l^*\right]_g  $.
\end{proof}

\subsubsection{Example ($\beta < 1$)}
Prosumer $l$ can contribute more on social welfare than prosumer $k$ in EV discharging, even when $l$'s utility is less than $k$'s utility. This case is shown by the following lemma.
\begin{lemma}
	If $\beta$ satisfies the following condition
	\begin{gather}
	\beta < \frac{ - \nabla F_k(0) + L_k \| \textbf{q}_k^{D*} \| -2 b_0(t) } { \nabla F_k(0) - L_k \| \textbf{q}_l^{D*} \| } < \frac{ b_0(t) } { - \nabla F_k(0)} \leq 1, \label{example2}
	\end{gather}
	then $\left[\boldsymbol{\lambda}_l^*\right]_g > \left[\boldsymbol{\lambda}_k^*\right]_r$, where these shadow prices are defined as the same as in Lemma \ref{ExLemma1}.
\end{lemma}

\begin{proof}
	We prove the inequality by showing the upper bound of $\left[\boldsymbol{\lambda}_k^*\right]_r$, $\left[\boldsymbol{\lambda}_k^*\right]_r^{upper}$ is less than the lower bound of $\left[\boldsymbol{\lambda}_l^*\right]_g$, $\left[\boldsymbol{\lambda}_l^*\right]_g^{lower}$. First, from \eqref{EVupper} we have
	\begin{gather}
	\left[\boldsymbol{\lambda}_k^*\right]_r^{upper} = \left[b_0(t)+ \frac{\partial F_k}{\partial q_{k,a}(t)}\left(0\right)\right]^+ \nonumber
	\\
	+ L_k\left(\|\overline{\textbf{q}}_k^{D*}\|+\|\textbf{q}_k^{D*}\|\right).\label{k_EVlower}
	\end{gather}
	Substituting \eqref{Lem_qUpper} to \eqref{k_EVlower}, we can obtain
	\begin{align}
	&\left[\boldsymbol{\lambda}_k^*\right]_r^{upper} = \left[b_0(t)+ \frac{\partial F_k}{\partial q_{k,a}(t)}\left(0\right)\right]^+ \nonumber
	\\
	&+ L_k \left\|- \frac{\nabla F_k( 0)}{L_k} - \frac{\textbf{b}_{0,k}}{ L_k} - \frac{\textbf{e}_k [\overline{\boldsymbol{\lambda}}_k^*]_r}{ L_k} \right\| +  L_k \|\textbf{q}_k^{D*}\| \nonumber.
	\end{align}
	Since we are considering the single-time-period case, $\frac{\partial F_k}{\partial q_{k,a}(t)}\left(0\right) =  \nabla F_k( 0) $ and $\textbf{b}_{0,k} = b_0(t)$. Also because the constraint for the shadow price $\left[\boldsymbol{\lambda}_k^*\right]_r$ is $-q_{k,a}^{D}\left(t\right)\leq 0$, we have $\textbf{e}_k =  -1$. Hence, $\left[\boldsymbol{\lambda}_k^*\right]_r^{upper}$ can be simplified to 
	\begin{align}
	\left[\boldsymbol{\lambda}_k^*\right]_r^{upper} = &\left[b_0(t)+  \nabla F_k( 0) \right]^+ + \left\|-  \nabla F_k( 0) \right.\nonumber
	\\
	&\left. - b_{0}(t)+ [\overline{\boldsymbol{\lambda}}_k^*]_r \right\| +  L_k \|\textbf{q}_k^{D*}\|.\label{k_EVlower_medium2}
	\end{align}
	By replacing subscript $l$ by $k$, setting $\beta = 1$ in \eqref{lamda_l_upper}, we can obtain $[\overline{\boldsymbol{\lambda}}_k^*]_r = \left[b_0(t)+ \nabla F_k( 0) \right]^+$. Therefore, \eqref{k_EVlower_medium2} becomes
	\begin{align}
	\left[\boldsymbol{\lambda}_k^*\right]_r^{upper} = &\left[b_0(t)+  \nabla F_k( 0) \right]^+ +\left\|- \nabla F_k( 0) - b_{0}(t)\right.\nonumber
	\\
	&\left.+ \left[b_0(t)+  \nabla F_k( 0) \right]^+ \right\| +  L_k \|\textbf{q}_k^{D*}\| .\label{k_EVlower_medium3}
	\end{align}
	Recall that $ - \nabla F_k(0) > 0$. Hence, from \eqref{example2} we know
	\begin{gather}
	b_0(t) + \nabla F_k(0) \leq 0,\nonumber
	\end{gather}
	thus \eqref{k_EVlower_medium3} can be simplified to
	\begin{gather}
	\left[\boldsymbol{\lambda}_k^*\right]_r^{upper} = - b_0(t) -  \nabla F_k( 0) + L_k \|\textbf{q}_k^{D*}\| .\label{l_EVlower_final}
	\end{gather}
	Next, we can obtain prosumer $l$'s shadow price lower bound by using \eqref{EVlower}, \eqref{scaledUtility} and \eqref{Lem_qUpper} (when use \eqref{Lem_qUpper}, let $\beta = 1$ and replace subscript $l$ by $k$). 
	\begin{align}
	\left[\boldsymbol{\lambda}_l^*\right]_g^{lower} = &\bigg[\left[b_0(t)+\beta \nabla F_k( 0)\right]^+ - \beta L_k \left\|- \frac{\nabla F_k( 0)}{L_k} \right.\nonumber
	\\
	&\left.- \frac{\textbf{b}_{0,l}}{ \beta L_k} - \frac{\textbf{e}_l [\overline{\boldsymbol{\lambda}}_l^*]_g}{ \beta L_k}\right\|- \beta L_k \|\textbf{q}_k^{D*}\|\bigg]^+ .\nonumber
	\end{align}
	Similar to the process we analyzed the coefficients $\textbf{b}_{0,k}$ and $\textbf{e}_k$ for $\left[\boldsymbol{\lambda}_k^*\right]_r^{upper}$, we can obtain $\textbf{b}_{0,l} = b_0(t)$, $\textbf{e}_l =  -1$ and $[\overline{\boldsymbol{\lambda}}_l^*]_g = \left[b_0(t)+  \beta \nabla F_k( 0) \right]^+$. Hence $\left[\boldsymbol{\lambda}_l^*\right]_g^{lower}$ can be simplified to
	\begin{align}
	\left[\boldsymbol{\lambda}_l^*\right]_g^{lower} = &\bigg[\left[b_0(t)+\beta \nabla F_k( 0)\right]^+ -\left\|- \beta \nabla F_k( 0) - b_{0}(t)\right.\nonumber
	\\
	&\left.+ \left[b_0(t)+ \beta  \nabla F_k( 0) \right]^+\right\|- \beta L_k \|\textbf{q}_k^{D*}\| \bigg]^+.\label{l_EVlower__medium}
	\end{align}
	From \eqref{example2}, we have
	\begin{gather}
	b_0(t) +  \beta \nabla F_k(0) > 0.\label{positiveUpperShadow_k}
	\end{gather}
	Having the above inequality, \eqref{l_EVlower__medium} can be simplified to
	\begin{gather}
	\left[\boldsymbol{\lambda}_l^*\right]_g^{lower} = \left[b_0(t)+  \beta \nabla F_k( 0) -  \beta L_k \|\textbf{q}_l^{D*}\|\right]^+.
	\end{gather}
	This is the lower bound of prosumer $l$'s shadow price $\left[\boldsymbol{\lambda}_l^*\right]_g$.\\
	Last, by \eqref{example2}, we know $ \nabla F_k(0) - L_k \| \textbf{q}_l^{D*} \| < 0$ and $ - \nabla F_k(0) + L_k \| \textbf{q}_k^{D*} \| -2 b_0(t) < 0$. Hence, from \eqref{example2}, we can obtain
	\begin{gather}
	\beta \left( \nabla F_k(0) - L_k \| \textbf{q}_k^{D*} \| \right) >   - \nabla F_k(0) + L_k \| \textbf{q}_k^{D*} \| -2 b_0(t).\nonumber
	\end{gather}
	Adding $b_0(t)$ on both sides of the above inequality, we have
	\begin{align}
	&b_0(t) + \beta \nabla F_k(0) - \beta L_k \| \textbf{q}_k^{D*} \| \nonumber
	\\
	&> - b_0(t) - \nabla F_k(0) + L_k \| \textbf{q}_k^{D*} \|. \nonumber
	\end{align}
	Since $- b_0(t) - \nabla F_k(0) + L_k \| \textbf{q}_k^{D*} \| = \left[\boldsymbol{\lambda}_k^*\right]_r^{upper}$ is nonnegative. Therefore, the left hand side of the above inequality is positive, which implies
	\begin{align}
	&\left[\boldsymbol{\lambda}_l^*\right]_g^{lower} = b_0(t) + \beta \nabla F_k(0) - \beta L_k \| \textbf{q}_k^{D*} \| \nonumber
	\\
	&> - b_0(t) - \nabla F_k(0) + L_k \| \textbf{q}_k^{D*} \| = \left[\boldsymbol{\lambda}_k^*\right]_r^{upper} .
	\end{align}
	Hence, we have $\left[\boldsymbol{\lambda}_l^*\right]_g \geq \left[\boldsymbol{\lambda}_l^*\right]_g^{lower} > \left[\boldsymbol{\lambda}_k^*\right]_r^{upper} \geq \left[\boldsymbol{\lambda}_k^*\right]_r$.
\end{proof}

The feasible areas of $\beta$ for the above two examples are depicted in Figure \req{fig:FA}. 
\begin{figure}
\center
\captionsetup{justification=centering}
\includegraphics[width=0.42\textwidth]{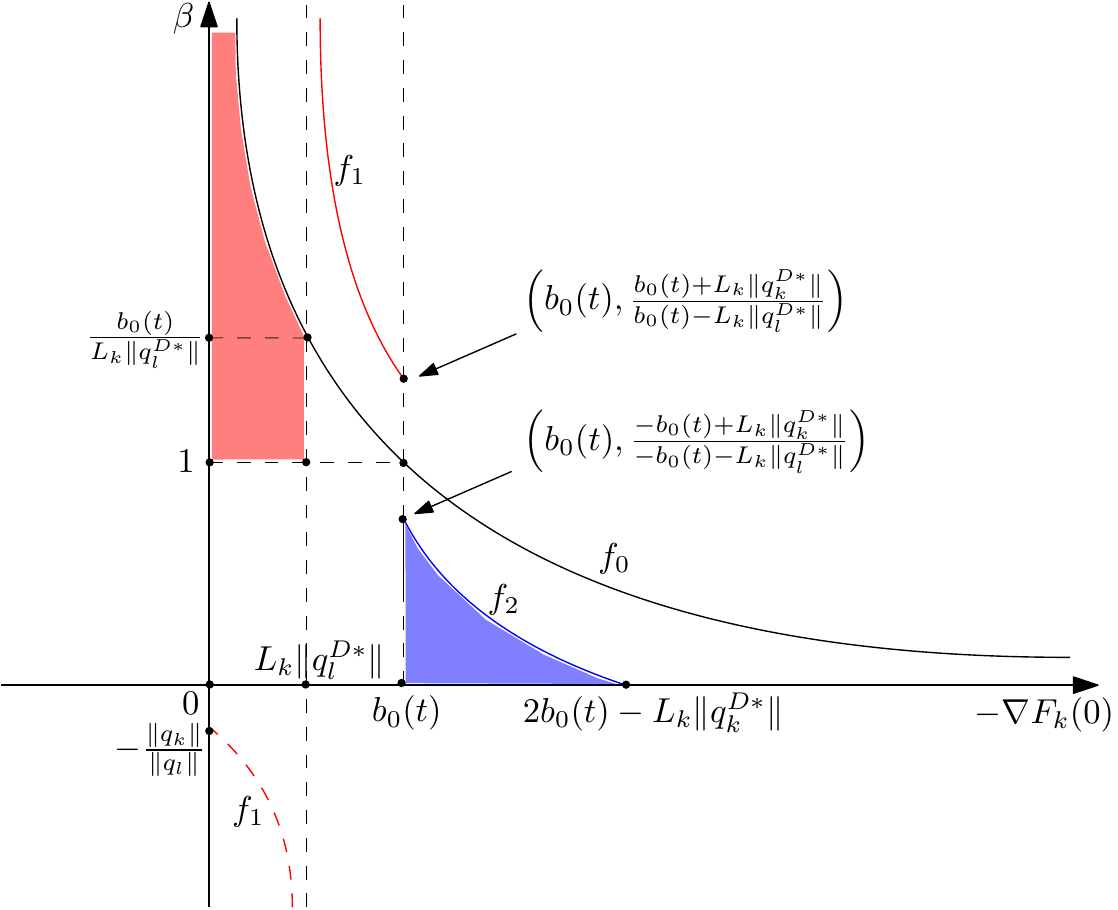}
	\caption{Feasible areas of $\beta$ for the above two examples}
	\label{fig:FA}
\end{figure}
The three functions $f_0, f_1$ and $f_2$ in terms of prosumer $k$'s initial utility increasing rate $- \nabla F_k(0)$, replaced by $x$ for convenience, are given below. 
$f_0\left(x\right) = \frac{b_0(t)}{x},\ 
f_1\left(x\right) = \frac{ x + L_k \| \textbf{q}_k^{D*} \| } { x - L_k \| \textbf{q}_l^{D*} \| },\ 
f_2\left(x\right) =  \frac{ x + L_k \| \textbf{q}_k^{D*} \| -2 b_0(t) } { -x - L_k \| \textbf{q}_l^{D*} \| }$.
In the above figure, the red area is when $\beta>1$, where we would expect prosumer $l$ to have a higher impact, but surprisingly $k's$ multiplier has a higher range, and the blue area is for $\beta<1$, where prosumer $k$ with lower utility function but a higher impact on social welfare.

\section{Numerical Studies}
We illustrate the usage of GSAA by considering the DR market consists of 2 prosumers, 1 utility company and 1 DSO. Here we use GSAA in an incremental way to estimate a prosumer's contribution potential, meaning that we increase the resource capacity of the general linear constraint for GSAA in small increments and then update GSAA iteratively instead of using GSAA only once for a large change of resource capacity. 

First, we test GSAA usage under quadratic settings in a 24-time-period DR market, where each prosumer has two appliances: energy storage (labeled by ES) and an EV (labeled by EV). Both appliances have quadratic net utility functions, where the coefficients follow the signs in Table \ref{Table Coef.}. Specifically, we let the second order coefficients $\hat{a}_{i,ES}\left(t\right)$ for two prosumers' ES net utility be randomly selected from $[-0.05, -0.02]$, here $\hat{a}_{1,ES}\left(t\right) = -0.02$ and $\hat{a}_{2,ES}\left(t\right) = -0.035$. The second order coefficients $\hat{a}_{i,EV}\left(t\right)$ and the first order coefficients $\hat{b}_{i,EV}\left(t\right)$ of EVs' net utility are randomly selected from $[-0.04, -0.01]$ and $[0.1, 0.5]$ respectively. In this experiment, $\hat{a}_{1,EV}\left(t\right) = -0.01$, $\hat{a}_{2,EV}\left(t\right) = -0.015$, $\hat{b}_{1,EV}\left(t\right) = 0.1$ and $\hat{b}_{2,EV}\left(t\right) = 0.2$. All other coefficients are 0 based on Table \ref{Table Coef.}. The utility company's production cost $b_{0}(t)$ is randomly chosen from $[0.2, 0.6]$ and here we have $b_{0}(t)=0.4$. For simplicity, we let the coefficients of prosumers' net utility remain unchanged over time and only include one constraint at the 1st time period in the model and make sure the constraint is tight at optimality. We conduct two experiments about using GSAA to compare the prosumers' marginal contribution on social welfare. The first one is on the net selling effects, which is shown in Figure \ref{fig:NetSell}. As we can see, both prosumers' contributions on social welfare increase as their allowed net-selling amounts (i.e., resource capacities) increase and each estimation given by GSAA is always an upper bound to the real contribution of the corresponding prosumer. For the same amount of allowed net-selling for both prosumers, the estimation by GSAA implies that prosumer 1's contribution potential (red lines with markers) is larger than prosumer 2's (blue lines with markers). This claim is verified by the fact that 1's real contribution (red line without marker) is higher than 2's (blue line without marker). We can also see that when the allowed net-selling amounts are small, the estimations by GSAA could relatively well reflects the prosumers' potential. However, as the allowed net-selling amount increases the estimation error becomes larger (This error will finally remain unchanged after exceeding some value). Hence, GSAA should be used carefully, if some prosumer's allowed net-selling amount is too large. 

The second simulation is on allowing EV discharging, which is shown in Figure \ref{fig:EVDis}. In this test, the estimation by GSAA is very close to the prosumers' real contribution on social welfare, even as the allowed EV discharging amount becomes large. Hence, in this test, the estimation by GSAA can help us quickly decide who has more contribution potential on social welfare when they are allowed to discharge EVs. 

\begin{figure}[H]
\begin{minipage}{0.23\textwidth}
\includegraphics[width=\textwidth]{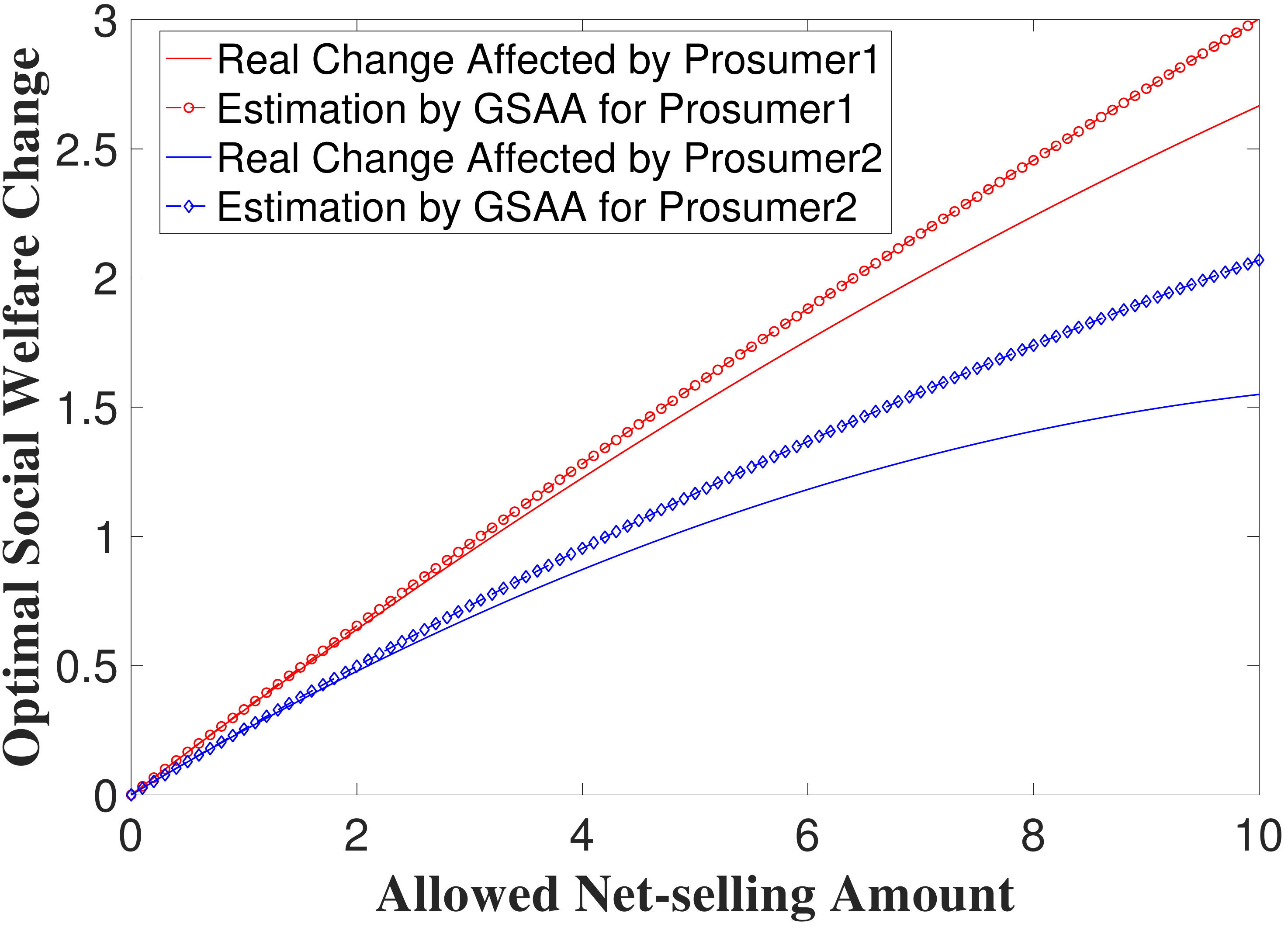}\captionof{figure}{\footnotesize{Prosumers' marginal contribution on social welfare with allowed net-selling amount, $M$.}}\label{fig:NetSell}
\end{minipage}
\hfill
\begin{minipage}{0.23\textwidth}            
               \includegraphics[width=\textwidth]{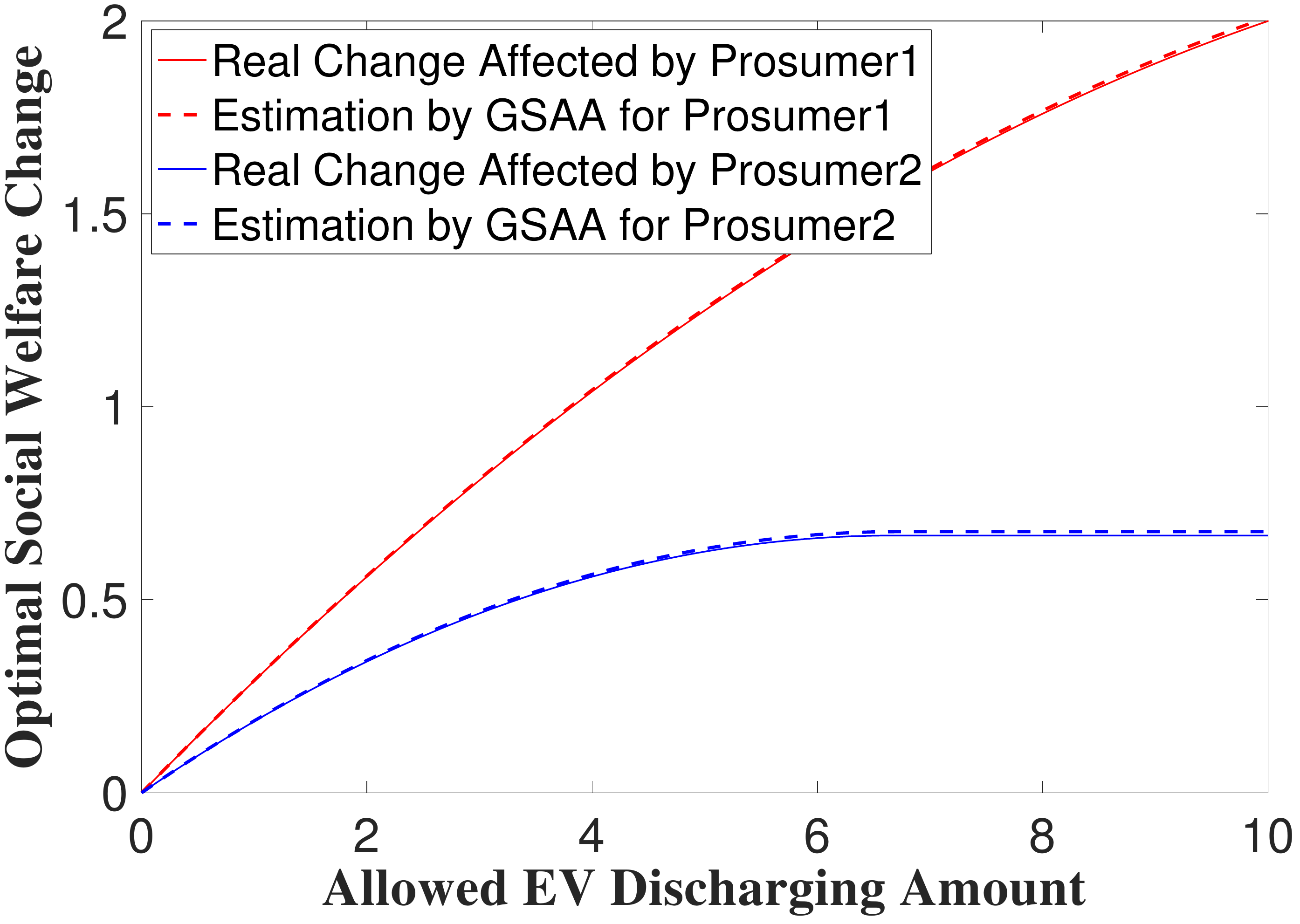}
               \captionof{figure}{\footnotesize{Prosumers' marginal contributions with change of EV initial utility increasing rate, $\hat{b}_{i,a}\left(t\right)$.}}\label{fig:EVDis}
\end{minipage}
\vspace{-0.2cm} 
\end{figure}

Next, we test GSAA performance under the general convex settings. The experiments are designed to simulate the case in Section \ref{Case Study}. For simplicity, we let each prosumer have only one appliance: an EV. We use quadratic functions as prosumers net utility functions since they satisfy the relevant assumptions. The coefficients of prosumer 1's EV net utility function are generated in the same way as in the previous simulation, here $\hat{a}_{1,EV}\left(t\right) = -0.01$ and $\hat{b}_{1,EV}\left(t\right) = 0.1$. The coefficients for prosumer 2's EV net utility function are prosumer 1's scaled up by a positive number $\beta$. Here we use $\beta = 2$. However, the above information are not accessible by GSAA. GSAA only has information about the degree of strong convexity, $\mu_i$, the Lipschitz constant, $L_i$, for prosumers' net utility functions, prosumers' operating points and initial utility increasing rates\footnote{The operating points and the initial utility increasing rates can be obtained since we have the specific coefficients of each prosumer's net utility function, even though GSAA does not have information about these coefficients.}. We set $\mu_1 = 0.018$, $L_1 = 0.022$, $\mu_2 = 0.036$ and $L_1 = 0.044$. The utility company's cost $b_{0}(t)$ is still chosen in the same way as in the previous simulation, here we have $b_{0}(t)=0.4$. The first experiment in Figure \ref{BoundsGSAA} shows the bounds of a prosumer 1's marginal contribution on social welfare estimated by GSAA with and without knowing 1's net utility functions, and also shows the real marginal contribution.
From the figure, we can see that both real optimal social welfare improvement and the estimation of optimal social welfare improvement given by the actual shadow price are bounded by the estimation based on the upper and lower bounds of shadow price. The figure also shows that as EV's initial utility increasing rate $\hat{b}_{1,EV}\left(t\right)$ changes, the estimation given by the upper bound of the shadow price works well when $\hat{b}_{1,EV}\left(t\right)$ is smaller than a threshold, the point where EV discharging stops contributing on the social welfare. However, when EV's initial utility increasing rate $\hat{b}_{1,EV}\left(t\right)$ is larger than the threshold, the estimation given by the upper bound of the shadow price begins to bounce up. The second experiment, shown in Figure \ref{TwoBoundsGSAA}, illustrates using the bounds obtained by GSAA, as indirect information, to compare two prosumers' potential on contributing social welfare. When the allowed EV-discharging amounts are smaller than 2 units and one's lower bound by GSAA is higher than the other's upper bound by GSAA, the bounds given by GSAA is sufficient for us to make a conclusion about who has more contribution potential on social welfare. For example, when prosumer 1 and prosumer 2 have 1 unit of allowed EV-discharging amount, by GSAA we can see that prosumer 1 has more contribution potential than prosumer 2 because 1's lower bound given by GSAA is higher than 2's upper bound given by GSAA, despite the fact that 2 has a higher utility function. This is also verified by the real contribution values. However, as the same as in the quadratic settings, estimations by GSAA should be used cautiously when the resource capacity change is too large, since the estimation error could also be large at this time and may lead to inaccurate prediction.

\begin{figure}
\begin{minipage}{0.23\textwidth}
\includegraphics[width=\textwidth]{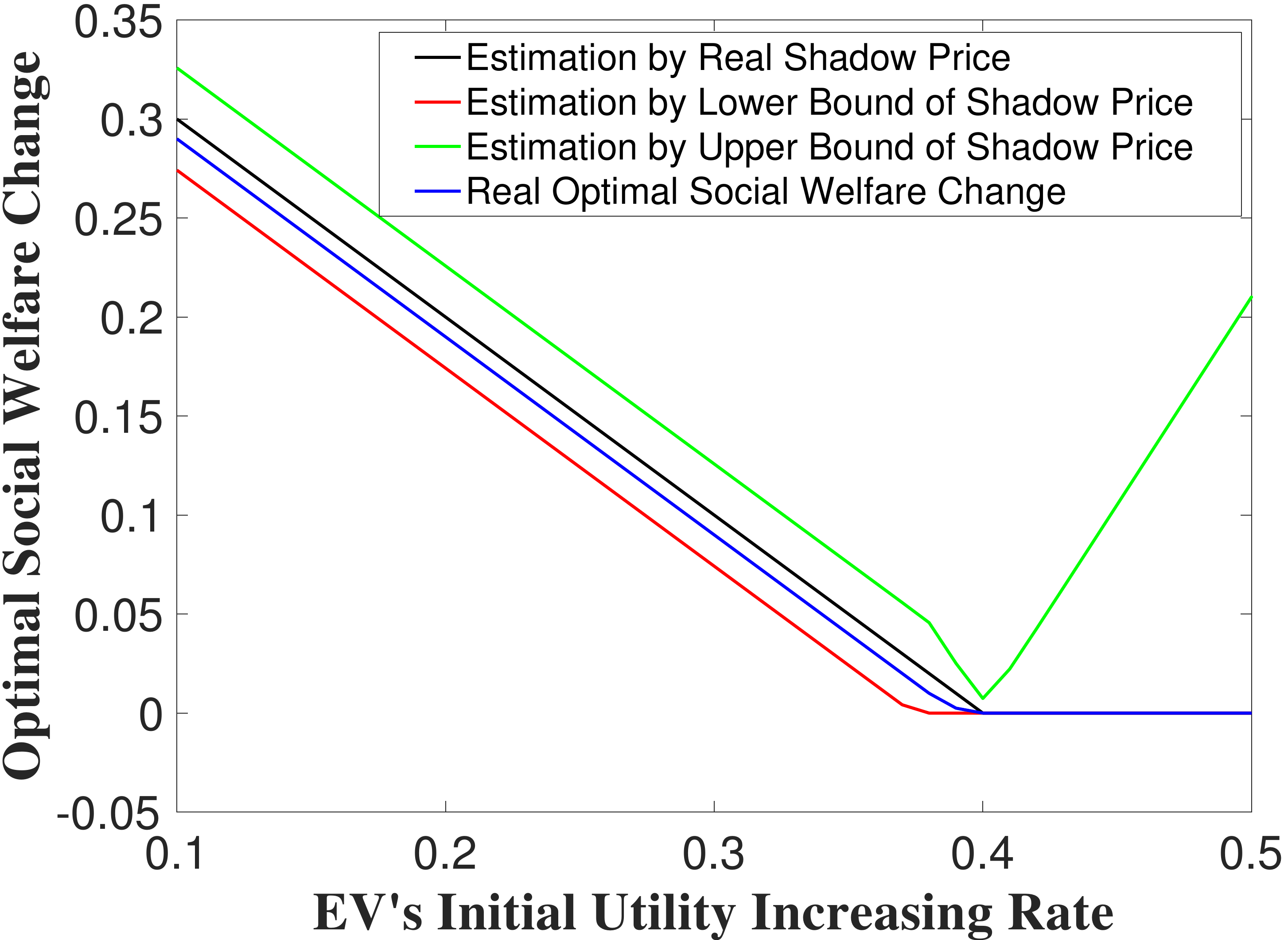}
	\caption{\footnotesize{Optimal social welfare improvement bounds with rise of prosumer 1's EV initial utility increasing rate, $\hat{b}_{1,EV}\left(t\right)$.}}
	\label{BoundsGSAA}
\end{minipage}
\hfill
\begin{minipage}{0.23\textwidth}            
               \includegraphics[width=\textwidth]{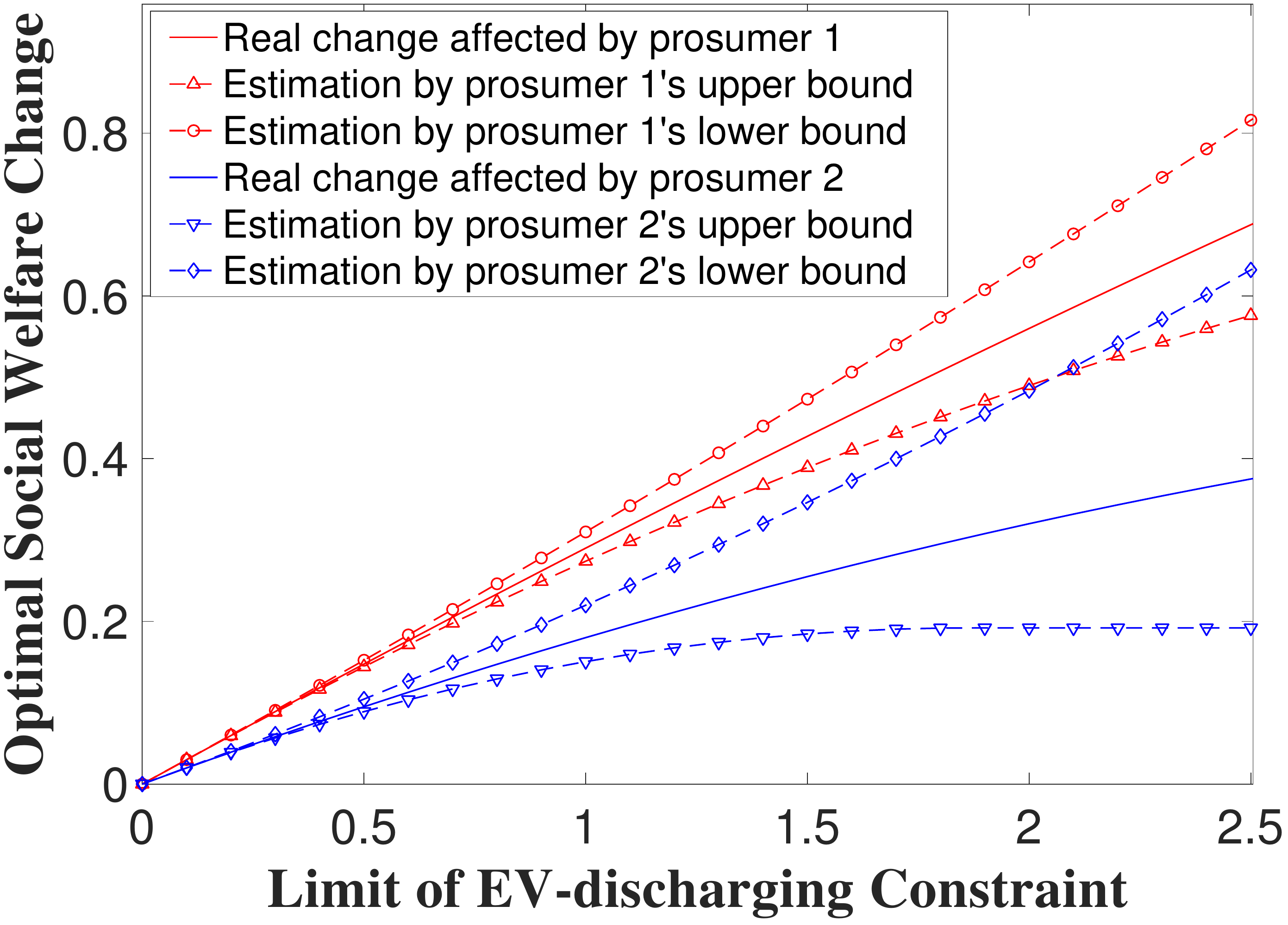}
               \caption{\footnotesize{Using shadow price bounds to compare contributions of two EVs' discharging on optimal social welfare improvements.}}
               \label{TwoBoundsGSAA}
\end{minipage}
\vspace{-0.2cm} 
\end{figure}

\section{Conclusion}\label{sec:concl}
In this paper, we consider an $H$-period DR market model consists of one distribution system operator, price-taking prosumers with various appliances, and an electric utility company. We show that there exists an efficient competitive equilibrium and an equivalence relationship between DSO's problem and the prosumers' and the utility company's individual problems. Based on the equivalence relationship and duality theory of convex optimization, we propose a general sensitivity analysis approach (GSAA) to analyze the effect of individual prosumers' contribution on social welfare, when their DR resource capacity is increased. We characterize closed-form shadow prices associated with the constraining resources, when the net utility function has quadratic form and provide the bounds of the shadow prices when the utility functions have general convex properties. The shadow prices can then be used to estimate improvement in social welfare and identify the most contributing prosumer(s). Thus, when the budgets for implementing DR (such as campaign/advertisement, device upgrade and installation, or the construction of infrastructure supporting EV discharging or net-selling) are limited, we could allocate the resources to the most contributing prosumers. Several applications of GSAA are provided, including enlarging a prosumer's AC comfort zone size, allowing a prosumer's EV to discharge and allowing a prosumer to net sell. We also provide several numerical studies on net-selling and EV discharging, which confirm our theoretical predictions.

\ifCLASSOPTIONcaptionsoff
  \newpage
\fi



\bibliographystyle{IEEEtran}

\bibliography{Paper1ref}
%
%
%

%





\end{document}